\documentclass{article}

\usepackage[utf8]{inputenc}
\usepackage[colorlinks = true,
            linkcolor = blue,
            urlcolor  = blue,
            citecolor = blue,
            anchorcolor = blue]{hyperref}

\usepackage{amsthm}
\usepackage{amssymb}
\usepackage{amsmath}
\usepackage{xcolor}
\usepackage{algorithm,caption}

\usepackage{cleveref}
\usepackage{graphicx}
\usepackage{scalerel}
\usepackage{enumitem}
\usepackage{natbib}
\usepackage{framed}

\newcommand{\DP}[1]{%
  \underset{#1}{\mathrel{\stretchto{\approx}{2mm}}}
}

\usepackage{algpseudocode}

\newcommand{\mynote}[2]{{\textcolor{#1}{ #2}}}
\definecolor{gray}{gray}{0.4}
\newcommand{\gray}[1]{\mynote{gray}{{\footnotesize #1}}}

\newcommand{\Lap}{\operatorname{\rm Lap}}
\newcommand{\Affspan}{\mathtt{Affspan}}
\newcommand{\NNN}{\mathcal N}

\usepackage{geometry}
\makeatletter
\theoremstyle{plain}
\newtheorem{theorem}{\protect\theoremname}
\newtheorem*{theorem*}{\protect\theoremname}
\newtheorem*{prop*}{\protect\theoremname}
\newtheorem{definition}{\protect\definitionname}
\theoremstyle{definition}

\theoremstyle{plain}
\newtheorem{lem}[definition]{\protect\lemmaname}
\newtheorem{claim}[definition]{\protect\claimname}
\newtheorem{cor}[definition]{\protect\corollaryname}
\newtheorem*{cor*}{\protect\corollaryname}
\newtheorem{remark}[definition]{\protect\remarkname}

\newtheorem{prop}[definition]{\protect\propname}

\newtheorem*{question*}{\protect\questionname}
\newtheorem*{assumption*}{\protect\assumptionname}
\newtheorem{obs}[definition]{\protect\observationname}
\makeatother

\newtheorem{innercustomthm}{Theorem}
\newenvironment{customthm}[1]
  {\renewcommand\theinnercustomthm{#1}\innercustomthm}
  {\endinnercustomthm}

\providecommand{\questionname}{Question}
\providecommand{\assumptionname}{Assumption}
\providecommand{\observationname}{Observation}
\providecommand{\corollaryname}{Corollary}
\providecommand{\definitionname}{Definition}
\providecommand{\lemmaname}{Lemma}
\providecommand{\claimname}{Claim}

\providecommand{\theoremname}{Theorem}
\providecommand{\exercisename}{Exercise}
\providecommand{\examplename}{Example}
\providecommand{\remarkname}{Remark}
\providecommand{\factname}{Fact}
\providecommand{\propname}{Proposition}

\newcommand{\remove}[1]{}
\newcommand{\ignore}[1]{}

\newcommand{\R}{\mathbb{R}}

\newcommand{\F}{\mathbb{F}}

\newcommand{\E}{\mathbb{E}}

\newcommand{\I}{\mathcal{I}}
\renewcommand{\O}{\mathcal{O}}
\newcommand{\A}{\mathcal{A}}
\newcommand{\eps}{\epsilon}
\newcommand{\Conv}{\mathtt{Conv}}
\newcommand{\Vol}{\mathtt{Vol}}
\newcommand{\diam}{\mathtt{diam}}

\newcommand{\x}{\boldsymbol{x}}

\newcommand{\N}{\mathbb{N}}
\newcommand{\Z}{\mathbb{Z}}
\renewcommand{\span}{\mathtt{span}}
\newcommand{\XXX}{{\cal X}}
\newcommand{\YYY}{{\cal Y}}
\newcommand{\EEE}{{\cal E}}
\newcommand{\JJJ}{{\cal J}}

\author{
Haim Kaplan\footnote{School of computer science, Tel Aviv University, and Google Research. Supported in part by ISF grant 1156/23 and the Blavatnik Research Foundation.}
\and
Yishay Mansour\footnote{School of computer science, Tel Aviv University, and Google Research. This project has received funding from the European Research Council (ERC) under the European Union’s Horizon 2020 research and innovation program (grant agreement No. 882396), by the Israel Science Foundation, the Yandex Initiative for Machine Learning at Tel Aviv University and a grant from the Tel Aviv University Center for AI and Data Science (TAD).}
\and
Shay Moran\footnote{Departments of Mathematics, Computer Science, and Data and Decision Sciences, Technion and Google Research. Robert J.\ Shillman Fellow; supported by ISF grant 1225/20, by BSF grant 2018385, by an Azrieli Faculty Fellowship, by Israel PBC-VATAT, by the Technion Center for Machine Learning and Intelligent Systems (MLIS), and by the European Union (ERC, GENERALIZATION, 101039692). Views and opinions expressed are however those of the author(s) only and do not necessarily reflect those of the European Union or the European Research Council Executive Agency. Neither the European Union nor the granting authority can be held responsible for them.}
\and
Uri Stemmer\footnote{School of computer science, Tel Aviv University, and Google Research. Supported in part by ISF grant 1419/24 and the Blavatnik Research Foundation.}
\and
Nitzan Tur\footnote{Work done while at Google Research.}
}

\title{On Differentially Private Linear Algebra}

\begin{document}
\date{November 5, 2024}
\maketitle

\begin{abstract}
We introduce efficient differentially private (DP) algorithms for several linear algebraic tasks, including solving linear equalities over arbitrary fields, linear inequalities over the reals, and computing affine spans and convex hulls. As an application, we obtain efficient DP algorithms for learning halfspaces and affine subspaces. Our algorithms addressing equalities are strongly polynomial, whereas those addressing inequalities are weakly polynomial. Furthermore, this distinction is inevitable: no DP algorithm for linear programming can be strongly polynomial-time efficient. 
\end{abstract}

\section{Introduction}

Consider a constraint satisfaction problem in which each constraint is treated as a sensitive data point containing private information (e.g., transaction details in a business, financial data of individuals, health metrics of patients, etc.). An important class of such problems arises in linear algebra, with the most basic example being a system of linear equations:
\begin{equation*}
\left\{
\begin{aligned}
    &a_{11}x_1 + a_{12}x_2 + \cdots + a_{1d}x_d = b_1 \\
    &a_{21}x_1 + a_{22}x_2 + \cdots + a_{2d}x_d = b_2 \\
    &\phantom{a_{m1}x_1 + a_{m2}x_2 + \cdots} \vdots &\phantom{+a_{md}x_d} \phantom{=}\\
    &a_{m1}x_1 + a_{m2}x_2 + \cdots + a_{md}x_d = b_m,
\end{aligned}
\right.
\end{equation*}
where the coefficients $a_{ij},b_i\in\mathbb{F}$ and $\mathbb{F}$ is a field. Another important example is linear programming, in which each constraint is a linear inequality and the field is $\mathbb{R}$.
    \begin{framed}
     \vspace{-2mm}
    \begin{center}
\emph{Can such tasks be solved while preserving differential privacy (DP)?} 
    \end{center}
    \vspace{-2mm}
    \end{framed}
The answer is, regrettably, no; the reason lies in the inherent instability of the existence of a solution. Indeed, consider two neighboring systems of linear equations (differing by only a single equation) such that one system is solvable, while the other is not. The presence of such neighboring systems indicates that the property of having a solution is not stable across minor changes to the set of equations, making it inherently impossible for an algorithm to both maintain differential privacy and produce solutions to feasible systems.

Thus, in order to maintain privacy, we must compromise on the solution quality. One natural compromise, which is applicable over any field \(\F\), is to produce a solution that satisfies \emph{many} equations. That is, given a (satisfiable) system of \(m\) equations, produce, in a differentially private manner, a solution to as many equations in the system as possible. Can this be performed by a DP algorithm? What is the largest number of equations the private solution can satisfy? We consider similar questions for other linear algebra tasks such as finding the solution space for a system of linear equalities, linear programming (here the field \(\F\) is assumed to be \(\F=\R\)), and other related tasks.

\subsection*{Problem Descriptions}
We now turn to formally define the linear algebraic tasks studied in this manuscript. Our focus is on constraint satisfaction problems where the solution space is either an affine subspace or a convex polyhedron.\footnote{Recall that the affine hull of a set of points consists of all linear combinations of these points with coefficients that sum to \(1\) (not necessarily positive). For example, the affine hull of two vectors is the line passing through them.} The goal is to devise a DP algorithm that outputs an approximation of an element (solution) within this space.

Affine subspaces and polyhedra can each be described in two dual ways. An affine subspace may be represented as either (i) the solution space of a given set of linear equations, or (ii) the affine hull of a given set of vectors. Similarly, a (compact) polyhedron may be represented as either (i) the solution space of a set of linear inequalities, or (ii) the convex hull of a given set of vectors. In the absence of privacy constraints and when working with representations (ii) — where we are given a set of vectors as input — identifying a point within their convex or affine hull is straightforward; we can simply output one of the input vectors. However, with privacy considerations, the landscape changes. Interestingly, if one requires privacy then the problems become non-trivial even under the second formulation.\footnote{That is, treating each point in the input sequence as sensitive data, the objective is to output a point within the convex or affine hull while preserving differential privacy.} In fact, as we will see there is a certain algorithmic 
equivalence between formulations (i) and (ii), allowing for an efficient transformation of DP algorithms suitable to one form to work on the other.

\subsubsection*{Linear Equalities and Affine Hulls}
\noindent
\begin{minipage}{0.6\textwidth}
    \begin{framed}
    \begin{center}
    \underline{\textbf{Primal (i)}}
    \end{center}
    \noindent\underline{\textbf{Input:}} system of linear equations \(\{\boldsymbol{a_i}\cdot \x = b_i\}_{i=1}^m \).
    
    \vspace{2mm}
    \noindent\underline{\textbf{Output:}} If the system is satisfiable,
    output a solution satisfying many equations (it is impossible to satisfy all equations while maintaining differential privacy).
    \end{framed}
\end{minipage}
\begin{minipage}{0.4\textwidth}
    \begin{framed}
    \begin{center}
    \underline{\textbf{Dual (ii)}}
    \end{center}
    \noindent\underline{\textbf{Input:}} a set of vectors \(\boldsymbol{v_1},\ldots \boldsymbol{v_m}\).
    
    \vspace{10mm}
    \noindent\underline{\textbf{Output:}} a point in their affine hull.
    \end{framed}
\end{minipage}%

\subsubsection*{Linear Inequalities and Convex Hulls}
\noindent
\begin{minipage}{0.6\textwidth}
    \begin{framed}
    \begin{center}
    \underline{\textbf{Primal (i)}}
    \end{center}
    \noindent\underline{\textbf{Input:}} system of linear inequalities \(\{\boldsymbol{a_i}\cdot \x \leq b_i \}_{i=1}^m\).
    
    \vspace{2mm}
    \noindent\underline{\textbf{Output:}} If the system is satisfiable,
    output a solution satisfying many inequalities (it is impossible to satisfy all equations while maintaining differential privacy).
    \end{framed}
\end{minipage}
\hfill
\begin{minipage}{0.4\textwidth}
    \begin{framed}
    \begin{center}
    \underline{\textbf{Dual (ii)}}
    \end{center}
    \noindent\underline{\textbf{Input:}} a set of vectors \(\boldsymbol{v_1},\ldots \boldsymbol{v_m}\).
    
    \vspace{10mm}
    \noindent\underline{\textbf{Output:}} a point in their convex hull.
    \end{framed}
\end{minipage}%

\section{Overview of Main Results}

In this section, we provide an overview of our main results, along with some remarks on key techniques and a discussion of corollaries related to algorithmic differential privacy.

\subsection{Linear Equalities (\Cref{sec:equalitiesmain})}
We introduce the first efficient differentially private (DP) algorithm for solving linear equalities and related tasks over arbitrary fields. For example, in \Cref{prop:linspan}, we design an \((\eps,\delta)\)-DP algorithm that, given an input set of vectors \(\{u_i\}\) in a \(d\)-dimensional vector space, outputs a set of vectors \(\{v_j\}\) in a differentially private manner such that the space \(V_{\text{output}} = \mathtt{span}(\{\boldsymbol{v_j}\})\) \emph{approximates} the space \(V_{\text{input}} = \mathtt{span}(\{\boldsymbol{u_i}\})\) in the following sense: \(V_{\text{output}} \subseteq V_{\text{input}}\) and contains all but at most 
\begin{equation}\label{eq:1}
k_{1}(d,\eps,\delta) = O\left(\frac{d^2}{\eps} \log\left(\frac{d}{\delta}\right)\right)
\end{equation}
of the input vectors. Additionally, we design algorithms for approximating affine hulls and solving systems of linear equations with similar guarantees: in each case, the bound \(k_1(d,\eps,\delta)\) applies to the number of unsatisfied equations or the points not in the output affine space; see \Cref{t:privateaffspan} and \Cref{t:privatelineq} for details.
We stress that these algorithms not only address fundamental tasks in linear algebra but also play an essential role in our approach to solving linear programs.

Previous algorithms for private subspace approximation either (i) assumed that the data comes from a multivariate Gaussian distribution \citep{SinghalS21,ashtiani2022private}, or (ii) only guaranteed utility for ``easy'' instances that adhere to certain ``niceness'' assumptions \citep{Tsfadia24}, or (iii) were computationally inefficient~\citep{SinghalS21}.\footnote{We remark that some of these works on private subspace approximation operate in a somewhat different model, where the goal is to find a subspace that {\em approximately} contains most of the data points.}

\subsubsection*{Highlight: Peeling Independent Sets (\Cref{alg:vector_partition})}
In the algorithms introduced above for approximating linear subspaces and affine hulls, we use a key technique of \textit{peeling independent sets}. Let \(\boldsymbol{u_1}, \ldots, \boldsymbol{u_n}\) be an input sequence of vectors in a linear space. The process involves peeling independent subsets from this sequence in a greedy manner: while the sequence is not empty, we scan it from left to right, accumulating an independent subset by adding each vector that remains independent of those selected so far, and then remove this subset from the sequence.

This simple independent-set peeling process forms the backbone of our DP algorithms for handling linear equations and subspaces. Essentially, it is stable in the sense that
 adding or removing a vector from the input sequence does not substantially alter the spans of the independent sets it produces (see \Cref{prop:stabind} for details). We believe this technique has potential beyond our current applications; for instance, it can be applied to any matroid with similar guarantees.

\subsection{Linear Programming and Convex Combinations (\Cref{sec:inequalitiesmain})}

In \Cref{thm:dp-lp}, we present the first efficient DP algorithm for solving general linear programs. Specifically, given a linear program, our algorithm outputs a solution that satisfies all but at most
\begin{equation}\label{eq:2}
k_{2}(d,\eps,\delta,U) = {\rm poly}(d,1/\eps,\log1/\delta,\log U)
\end{equation}
of the constraints, where \( U \) is an upper bound on the input integers defining the linear program, and \(\eps\) and \(\delta\) are the privacy parameters.
Our approach to linear programming builds on and extends an algorithm by \cite*{DS08}, which is based on the perceptron method. Specifically, in \Cref{thm:vempala}, we first analyze a DP version of \cite{DS08}'s algorithm, and then in \Cref{thm:dp-lp}, we extend it to handle arbitrary systems, including those with zero margin (whereas the original algorithm applies only to systems with a positive margin). Notably, addressing general systems relies on the linear system solvers discussed in the previous section. 

Previous work on this problem provided LP solvers that either (i) only \textit{approximately} satisfy \textit{most} constraints, meaning that most constraints may not be fully satisfied~\citep*{HsuRRU14}, or (ii) \textit{exactly} satisfy \textit{most} constraints but are computationally inefficient~\citep*{KaplanMST20}.  We note that it is possible to adjust the approximation parameter in the algorithm by \cite{HsuRRU14}, based on \(U\), to achieve an exact satisfaction of most constraints, but this would render their algorithm’s running time exponential. In contrast, our algorithm is both efficient and \textit{exactly} satisfies most constraints (as discussed in the introduction, it is impossible to satisfy all constraints, even approximately).

Another related work by \citet*{MSVV21} considered a similar model but assumed that only the vector of biases is private (i.e., only the scalars on the right-hand side of each inequality in the linear program). Under this assumption, they devised efficient DP algorithms. Our algorithms, however, apply in the more challenging setting where the entire constraint is treated as a private data point, thus addressing an open question left by \cite{MSVV21} for future research.

\medskip

We also consider the dual problem of privately finding a point in the convex hull. This problem was introduced by \cite*{BeimelMNS19} and was then studied by several followup papers including \cite{KaplanSS20}, \cite{GaoS21}, and \cite{BenEliezerMZ22}. 
All these works presented {\em computationally inefficient} algorithms for the problem.
Our techniques allow us to obtain the first computationally efficient algorithm, as stated in \Cref{thm:pich}. Our algorithm returns a point in the convex hull provided that the number of input points is at least
\begin{equation}\label{eq:3}
k_{3}(d,\eps,\delta,U) = {\rm poly}(d,1/\eps,\log1/\delta,\log U).
\end{equation}

\subsubsection*{Highlight: DP Linear Programming = DP Convex Combinations}
Our algorithm for finding a point in the convex hull, which privately outputs a convex combination of its input points, is derived using a reduction to the DP linear programming solver, treating it as a black box. Interestingly, there is also a reduction in the reverse direction: given black-box access to a DP algorithm that outputs a point in the convex hull, it is possible to construct a DP linear program solver. Thus, in the context of differential privacy, we achieve an \textit{equivalence} between these problems. In contrast, without privacy constraints, outputting a point in the convex hull is trivial, whereas solving LPs is far from trivial and has inspired a vast body of research and algorithms.

The equivalence (i.e., the reductions between finding a point in the convex hull and solving LPs) relies on an adaptation of the ellipsoid method:

\begin{itemize}
    \item \textbf{DP Convex Combinations \(\rightarrow\) DP Linear Programming}: In this reduction, we follow the ellipsoid method, maintaining an ellipsoid that contains the solution space. At each iteration, we check whether the center of the ellipsoid satisfies most of the constraints; if not, we apply the DP point-in-convex-hull algorithm to obtain a convex combination of the unsatisfied constraints, which itself forms an unsatisfied constraint. We then use this combination to bisect the ellipsoid and proceed to the next iteration.

    \item \textbf{DP Linear Programming \(\rightarrow\) DP Convex Combinations}: Here, we again follow the ellipsoid method, maintaining an ellipsoid that contains all input points. At each step, we check whether the center of the ellipsoid is a convex combination of the input points; if it is, we output it, and if it is not, we use the DP LP solver to find a hyperplane separating the ellipsoid’s center from the input points. We then bisect the ellipsoid and proceed to the next iteration.
\end{itemize}

\subsubsection*{Highlight: Weak vs.\ Strong Polynomial Time}
A longstanding open problem in optimization is whether there exists a linear programming (LP) solver that operates in strong polynomial time—that is, one whose runtime depends only on the number of variables and constraints, not on the bit-length of input values (assuming unit cost arithmetic) (see e.g.,~\cite*{grotschel2012geometric}). Perhaps surprisingly, in the realm of differentially private (DP) algorithms, this question is resolved: we show that weakly polynomial-time DP LP solvers exist, but prior results imply that strongly polynomial-time DP LP solvers do not, even in 1D. Indeed, in 1D, an LP reduces to finding a point in the intersection of rays. An even easier problem is the \emph{interior point problem}: given \( n \) input points on a line, find a point between the maximum and minimum values. It is known that any DP algorithm for this problem requires \( n \geq \log^\star(U) \) input points, where \( U \) bounds the absolute value of the input (integer) values~\citep*{BNSV15}. This implies that even for this simple case, a strongly polynomial-time algorithm cannot exist.

\subsection{Learning}
As applications, our results yield efficient PAC learning algorithms for subspaces (over arbitrary vector spaces) and for half-spaces (over Euclidean spaces). This follows from the fact that finding a consistent subspace reduces to solving a system of linear equalities, while finding a consistent half-space reduces to finding a point in the convex hull~\citep*{BeimelMNS19}. In \Cref{sec:pacsubspace} we provide a detailed description of the results concerning learning subspaces.

\subsection{Open Questions}
Our work naturally raises some open questions. First, it would be interesting to obtain tighter bounds on the number of unsatisfied constraints in systems of linear equalities (\Cref{eq:1}) and in linear programs (\Cref{eq:2}), as well as on the minimum input size required for finding a point in the convex hull (\Cref{eq:3}).\footnote{The exact power of $d$ in these equations is at least $5$ in all cases, even if the linear program has positive margin/roundness. Specifically for finding a point in the convex hull it is $17$. Our intent with this open question is regarding {\em computationally efficient} algorithms, without any distributional/data assumptions.} Notably, a linear dependence on the dimension \(d\) for these problems is necessary. This lower bound follows from a basic learning-theoretic consideration, leveraging the intrinsic connection between differential privacy and learning.

To illustrate, suppose for contradiction that the dependence on \(d\) in the number of unsatisfied points in \Cref{t:privateaffspan} could be improved to, say, \(d^{0.99}\). Via \Cref{t:privgen} and reasoning similar to that in the proof of \Cref{t:privsubspacelearn}, this would imply a PAC learning algorithm for subspaces in \(\F^d\) with an error rate scaling as \(\tilde{O}(d^{0.99}/m)\).  However, this would contradict the lower bound of \(\Omega(d/m)\) on the error rate for this task, which holds even without privacy constraints, as \(d\) is the VC dimension of \(d\)-dimensional subspaces. Similar arguments imply that a linear dependence on \(d\) is necessary also in the other tasks, such as those in \Cref{prop:linspan} and \Cref{t:privatelineq}. Closing the gap between these basic lower bounds and our upper bounds remains an open problem.

Additionally, it would be interesting to find more direct algorithms for the point-in-convex-hull problem. Our current approach relies on an LP solver, combined with a reduction from point-in-convex hull to LP. A direct solution for this problem would be valuable and could potentially offer new insights into linear programming through the reduction in the opposite direction.

Our presentation assumes numeric operations with arbitrary precision. We use continuous Laplace noises, and the ellipsoid algorithm and the LP algorithm of \cite{DS08}, which we augment, take square roots. We believe that it is possible to implement our algorithms using finite precision, while keeping the same running time, by applying standard techniques. This requires an appropriate discretization of the noises (as in \citep{BalcerV19}), an appropriate perturbation of the Ellipsoids (as described in \citep{grotschel2012geometric}), and an appropriate discretization of the algorithm of \cite{DS08}.

\section{Subspaces and Linear Equalities}\label{sec:equalitiesmain}
In this section, we develop and analyze differentially private algorithms for solving linear algebraic problems over arbitrary fields \(\F\). In Section~\ref{sec:equalities}, we address tasks such as solving systems of linear equations and computing linear and affine spans. An application of these results, achieving efficient differentially private learning algorithms for linear and affine subspaces, is presented in Appendix~\ref{sec:pacsubspace}.%

\subsection{Synthetic Linear Equalities and Affine Hulls Generation}\label{sec:equalities}

We solve a more general task of differentially private sanitization. The objective is to generate and publish synthetic output data that closely approximates the original input. For instance, in the context of solving linear equations, our sanitizer will receive a system of linear equations as input and, while maintaining differential privacy, will produce a corresponding output system of equations. This output system satisfies that its solution space approximates the solution space of the original input system.

\begin{remark}\label{rem:effectivedim}
    In all the results presented in this section, the parameter \(d\) which represents the dimension, can be substituted with the effective dimension of the space spanned by the input vectors. For instance, our differentially private algorithms for handling systems of linear equations retain their guarantees even if the number of variables is substantially large, provided that the space spanned by the coefficients of the equations is low dimensional.

\end{remark}

\subsubsection*{Sanitization of Linear Systems and Subspaces}
\noindent
\begin{minipage}{0.5\textwidth}
    \begin{framed}
    \begin{center}
    \underline{\textbf{Primal}}
    \end{center}
    \noindent\underline{\textbf{Input:}} linear system \(\mathcal{I}=\{\boldsymbol{a_i}\cdot \x = b_i \}_{i=1}^m\).
    
    \vspace{2mm}

    \noindent\underline{\textbf{Output:}}
    Generate in a differentially private manner a linear system \(\mathcal{O}=\{\boldsymbol{c_j}\cdot \mathbf{x} = d_j \}_{j=1}^k\) that approximates the input linear system \(\mathcal{I}\) (see \Cref{t:privatelineq}).

    \end{framed}
\end{minipage}
\hfill
\begin{minipage}{0.5\textwidth}
    \begin{framed}
    \begin{center}
    \underline{\textbf{Dual}}
    \end{center}
    \noindent\underline{\textbf{Input:}} a sequence \(\{\boldsymbol{u_i}\}_{i=1}^m\subseteq \F^d\).
    
    \vspace{2mm}
    \noindent\underline{\textbf{Output:}} Generate in a differentially private manner a sequence \(\{\boldsymbol{v_j}\}_{j=1}^k\subseteq \F^d\) such that \(\mathtt{Affine}(\{v_j\})\) approximates \(\mathtt{Affine}(\{\boldsymbol{u_i}\})\) (see \Cref{t:privateaffspan}).
    \end{framed}
\end{minipage} 

\subsubsection{Stable Partitions}
The following procedure is a central key component in all our algorithms. It gets a sequence of vectors as input and produces and outputs a partition of the input vectors into multiple sets, each of which is linearly independent. A key property of the procedure is its stability, a feature that is crucial for obtaining the desired privacy guarantees in our final algorithms.

\begin{algorithm}
\caption{Stable Partition to Linearly Independent Sets}
\label{alg:vector_partition}
\textbf{Input:} A sequence of vectors $\{\boldsymbol{v_i}\}_{i=1}^m\in\mathbb{F}^d$, where $\boldsymbol{v_i}\neq \boldsymbol{0}$ for all $i$.\\
\textbf{Output:} A partition of $\{\boldsymbol{v_i}\}$ into sets $\{A_j\}$ each of which is linearly independent. %

\begin{algorithmic}[1]
\State Initialize an empty list $\mathcal{P}$ to store the independent sets.
\State Initialize $i = 1$.
\While{input sequence is not empty}
    \State Initialize an empty set $A_i$.
    \For{each vector $\boldsymbol{v_j}$ in the input sequence}. \Comment{vectors are traversed in their original order}
        \If{$A_i\cup\{\boldsymbol{v_j}\}$ is independent}.
            \State Add $\boldsymbol{v_j}$ to $A_i$.
        \EndIf
    \EndFor
    \State Add $A_i$ to the list of independent sets $\mathcal{P}$.
    \State Remove vectors in $A_i$ from the input sequence.
    \State $i = i + 1$.
\EndWhile
\State \textbf{Output} $\mathcal{P}$.
\end{algorithmic}
\end{algorithm}

\begin{prop}\label{prop:stabind}
Given an input sequence $\{\boldsymbol{v_i}\}_{i=1}^m \subseteq \F^d$ to \Cref{alg:vector_partition}, let $\{A_t\}_{t=1}^k$ be the output partition. Then, there exist at most $d$ nested subspaces $V_1\subseteq \ldots \subseteq V_d$ such that each $\mathtt{span}(A_j)$ is one of these subspaces.
Furthermore, for a subspace $V$ define the basis-count $n(V)$ as the number of sets $A_i$ whose span equals a subspace $V$. 
Then, the mapping \(\{\boldsymbol{v_i}\} \mapsto n(\cdot)\)  has \(\ell_\infty\)-sensitivity $1$ and $\ell_1$-sensitivity $2$;
that is, for any neighboring input sequence $\{\boldsymbol{v_i}\}_{i=1}^m\setminus\{\boldsymbol{v_j}\}$ with basis-count $n_{-j}(\cdot)$, we have 
\[\max_{V}\lvert n(V) -n_{-j}(V)  \rvert\leq 1 \quad\text{and}\quad \sum_V\lvert n(V) -n_{-j}(V)  \rvert\leq 2.\]

\end{prop}
\begin{proof}
We start with first part. Let $U_t$ denote the subspace spanned by the vectors that remain in the input sequence after iteration $t-1$.
Observe that $U_1\supseteq U_2\supseteq\ldots$ and that $A_t$ is a basis for $U_t$.\footnote{That $A_t$ is a basis for $U_t$ is a consequence of the matroid property of independent sets: if $I_1,I_2$ are two linearly independent sets with $\lvert I_2\rvert > \lvert I_1\rvert$, then there exists $\boldsymbol{v}\in I_2$ such that $I_1\cup\{\boldsymbol{v}\}$ remains linearly independent.} 
Further observe that $U_t\supsetneq U_{t+1}$ if and only if $\dim(U_t) > \dim(U_{t+1})$. Thus, there are at most $d$ distinct subspaces among the $U_t$'s.

We now prove the second part by induction on the number of iterations in the ``while'' loop (lines 3-13); that is, the number of independent sets in the partition outputted by the algorithm.
Imagine we run two parallel applications of \Cref{alg:vector_partition}, once on the input sequence \(\{\boldsymbol{v_i}\}_{i=1}^m\) and the other on the sequence $\{\boldsymbol{v_i}\}_{i=1}^m\setminus\{\boldsymbol{v_j}\}$.
Both runs are identical up to the iteration where~$\boldsymbol{v_j}$ is selected to the independent set in the first run; {note that such an iteration exists because in each iteration of the while loop at least one vector is removed from the input sequence, and thus it is eventually empty.}
Let the index of this iteration be denoted by \(t\).
Let $A_t,A_t^{-j}$ denote the linearly independent sets produced in iteration~$t$ by the first and second runs of the algorithm respectively.
We distinguish between two cases:
\begin{enumerate}
    \item $A_t=A_t^{-j}\cup\{v_j\}$,
    \item $A_t\neq A_t^{-j}\cup\{v_j\}$.
\end{enumerate}
In the first case, the two runs are identical from iteration $t+1$ and onwards. Thus, $n(\cdot)$ and $n_{-j}(\cdot)$ differ only on the two subspaces $\mathtt{span}(A_t)$ and $\mathtt{span}(A_t^{-j})$, and the difference equals $1$ as required.

In the second case, there must be a vector $\boldsymbol{v_\ell}$ with $\ell>j$ such that $\boldsymbol{v}_\ell$ is added to $A_t^{-j}$ but not to $A_t$.
Taking the smallest index $\ell$ of such a vector, let $B=A_t\cap\{\boldsymbol{v_i} : i<\ell\}$, \(B_{-j}=A_t^{-j}\cap\{\boldsymbol{v_i} : i<\ell\}\) be the independent sets accumulated in the two runs of the algorithm when reaching the vector $\boldsymbol{v_\ell}$. Thus, the set $B\cup\{\boldsymbol{v}_\ell\}$ is linearly \underline{dependent}, the set $B_{-j}\cup\{\boldsymbol{v}_\ell\}$ is linearly \underline{independent}, and $\mathtt{span}(B)=\mathtt{span}(B_{-j}\cup\{v_\ell\})$.
Hence, the remainder of the $t$'th iteration is identical in the two runs and we have $\span(A_t)=\span(A_t^{-j})$. Therefore, by the end of the $t$'th iteration the counts $n(\cdot)$ and $n_{-j}(\cdot)$ match, and the remaining input sequences differ in a single entry -- the vector $v_\ell$. In other words, we arrive at an identical state as prior to the $t$'th iteration, and so applying induction on the subsequent iterations completes the proof.
\end{proof}

\subsubsection{Synthetic Linear Spans}
The following algorithm provides a differentially private (DP) approximation of the linear span of a given sequence of input vectors. It is a key subroutine in our subsequent algorithms for affine spans and solving linear equations.

\begin{algorithm}
\caption{Private Linear Span}\label{alg:linspan}
\textbf{Input: }a sequence \(\{\boldsymbol{u_i}\}_{i=1}^\ell\subseteq \F^d\).\\
\textbf{Output:} a sequence \(\{\boldsymbol{v_j}\}_{j=1}^k\subseteq \F^d\)
     such that \(\mathtt{span}(\{\boldsymbol{v_j}\})\) approximates \(\mathtt{span}(\{\boldsymbol{u_i}\}) \) (see \Cref{prop:linspan})
\begin{algorithmic}[1]
\State Apply \Cref{alg:vector_partition} on the sequence of vectors $\{\boldsymbol{u_i}\}_{i=1}^\ell$.
\State Let $\mathcal{P}$ denote its output partition, and for $k\leq d$ let $m(k)$ denote the number of sets in $\mathcal{P}$ whose size is $k$.
\State Initialize $\theta = \frac{16}{\varepsilon}\ln(100d/\delta) +  \mathtt{Lap}(2/\eps)$.
\For{\(k=d\ldots 1\)}. \Comment{higher-dimensional subspaces are traversed first}
    \If{$m(k) + \mathtt{Lap}(4/\eps) > \theta$}      
    \State Exit the for loop.
    \EndIf
\EndFor
\State If $m(k)>0$ then set $U$ to be the subspace spanned by the sets in $\mathcal{P}$ of size $k$; else set $U=\{\boldsymbol{0}\}$.
\State {\bf Output} a canonical basis of the space $U$ (see Remark~\ref{rem:canonB}).
\end{algorithmic}
\end{algorithm}

\begin{remark}\label{rem:canonB}
A canonical basis of $U$ of dimension $q$ can be derived by initially selecting any basis $B$ of $U$ and constructing a $q\times d$ matrix $M$, with the vectors of~$B$ as its rows. Let $A$ be the $q\times q$ sub-matrix of $M$ containing the $q$ left-most independent columns of $M$ in lexicographical order (call these columns $c_1,\dots,c_q$). Let $T=A^{-1}\, M$, and note that columns $c_1,\dots,c_q$ of $T$ form the identity $q\times q$ matrix. 
The rows of $T$ are a basis for $U$, as they are independent (as they contain $I$ as a sub-matrix) and were obtained by row operations on the independent vector set $B$ (and so span the same space as $B$).
Furthermore, this basis is canonical. That is, this process yields the same matrix $T$, regardless of the initial choice of basis $B$ for $U$.
To see this, note that the only vector in the span of the rows of $T$ that is zero in all coordinates $c_1,\dots,c_q$ except a single coordinate which is 1, is the corresponding row itself.
\end{remark}

\begin{customthm}{A}[Private Linear Span]\label{prop:linspan}
    \Cref{alg:linspan} is $(\eps,\delta)$-DP and satisfies the following with probability at least $1-\delta$.
Let $V_{input}=\mathtt{span}(\{\boldsymbol{u_i}\})$ denote the span of the input sequence and let $V_{output}=\mathtt{span}(\{\boldsymbol{v_j}\})$ denote the span of the output sequence. Then,
\begin{enumerate}
\item $V_{output}\subseteq V_{input}$.
\item All but $O(\frac{d^2}{\eps}\log(d/\delta))$ of the input vectors $\boldsymbol{u_i}$ satisfy $\boldsymbol{u_i}\in V_{output}$.
\end{enumerate}
\end{customthm}
\begin{proof}
By \Cref{prop:stabind}, the 
sensitivity of each of the counts in the vector $m(\cdot)$,
(computed by an application of \Cref{alg:vector_partition} to $S$) is $1$.
Therefore, we can think about the for loop as an application of the Sparse Vector Technique (SVT) (see \cite{DR14}) to this sequence of counts (checking, in turn, if they are above threshold $\theta$).
This SVT application outputs the index $k$.

Let \(S=\{\boldsymbol{u_i}\}_{i=1}^\ell\) and \(S_{-j}=\{\boldsymbol{u_i}\}_{i=1}^\ell\setminus\{\boldsymbol{u_j}\}\) be two neighboring input sequences.
Let $U$ be a canonical description of some subspace. Let $n(U)$ the number of 
occurrences of $U$ in the partition computed by \Cref{alg:vector_partition} on $S$ and by $n_{-j}(U)$ the number of occurrences of $U$ in the partition computed by \Cref{alg:vector_partition} on $S_{-j}$.

If  $n(U)=0$ and $n_{-j}(U)=0$ then both on $S$ and on $S_{-j}$ we get $U$ as the output with probability $0$.

If $n(U)>0$ and $n_{-j}(U)>0$
we get $U$ as the output on $S$ iff 
$k=\mathtt{dim}(U)$ when the input is $S$, and we get 
$U$ as the output on $S_{-j}$ iff 
$k=\mathtt{dim}(U)$ when  the input is $S_{-j}$.
Since SVT is $(\eps,0)$-DP then  we conclude that
$p(U \mid S)\DP{\eps,0} p(U \mid S_{-j})$.\footnote{
For an event $E$, we say that $p(E \mid S)\DP{\eps,\delta} p(E \mid S_{-j})$
iff  $p(E \mid S)\leq \exp(\eps)\cdot p(E \mid S_{-j}) + \delta$ and
$p(E \mid S_{-j})\leq \exp(\eps)\cdot p(E \mid S)+ \delta$. }

If $n(U)>0$ and $n_{-j}(U)=0$ then 
by the sensitivity property in \Cref{prop:stabind}, it must be the case that  $n(U)=1$ (and therefore $m(\mathtt{dim}(U))=1$). 
As before, to get $U$ as the output on $S$, the SVT (on $S$) has to output $k=\mathtt{dim}(U)$.
We get $k=\mathtt{dim}(U)$ as the output
of the SVT on $S$ with probability
at most
{\small
\begin{equation} \label{eq:deltaover4}
\begin{aligned}
 \Pr\Bigl[1 + \mathtt{Lap}\Bigl(\frac{4}{\eps}\Bigr) \ge \mathtt{Lap}\Bigl(\frac{2}{\eps}\Bigr) + \frac{16}{\varepsilon}\ln\frac{100d}{\delta}\Bigr]
 &\leq 
 1-\Pr\Bigl[\Bigl\lvert \mathtt{Lap}\Bigl(\frac{4}{\eps}\Bigr)\Bigr\rvert <  \frac{8}{\varepsilon}\ln\frac{100d}{\delta} - 1\Bigr]\cdot \Pr\Bigl[\Bigl\lvert \mathtt{Lap}\Bigl(\frac{2}{\eps}\Bigr)\Bigr\rvert <  \frac{8}{\varepsilon}\ln\frac{100d}{\delta} - 1 \Bigr]
 \\ %
    &\leq 1-\Pr\Bigl[\Bigl\lvert \mathtt{Lap}\Bigl(\frac{4}{\eps}\Bigr)\Bigr\rvert <  \frac{4}{\varepsilon}\ln\frac{100}{\delta}\Bigr]\cdot \Pr\Bigl[\Bigl\lvert \mathtt{Lap}\Bigl(\frac{2}{\eps}\Bigr)\Bigr\rvert <  \frac{4}{\varepsilon}\ln\frac{100}{\delta}\Bigr] \\
    &\leq 1- \Bigl(1-\frac{\delta}{100}\Bigr)^2 \leq \frac{\delta}{4},
\end{aligned}
\end{equation}
}

where we used above the following concentration inequality for Laplace distributions: 
\begin{equation}\label{eq:conlap}
  (\forall b, \Delta>0):\Pr[\lvert \mathtt{Lap}(b)\rvert > \Delta]\leq \exp(-\Delta/b).
\end{equation}

The probability of getting $U$ when the input is $S_{-j}$ is zero.
So we get that
$p(U \mid S)\DP{0,\delta/4} p(U \mid S_{-j})$.

If $n(U)=0$ and $n_{-j}(U)>0$ then 
the probability of getting $U$  
when the input is $S$ is $0$. A similar reasoning as above shows that this probability is at most $\delta/4$ when the input is $S_{-j}$. So we also get that 
$p(U \mid S)\DP{0,\delta/4} p(U \mid S_{-j})$.

Finally, consider the probability that the output is $U=\{\boldsymbol{0}\}$. For this to happen (either on $S$ or on $S_{-j}$) the SVT should output an index $k$ such that
$m(k)=0$. By \Cref{prop:stabind}
the number of indices $k$ for which
$m(k)=0$ and $m_{-j}(k)=1$ or vice versa is at most $2$. Therefore,
a calculation similar to the one in \Cref{eq:deltaover4} shows
that 
$p(\{\boldsymbol{0}\} \mid S)\DP{0,\delta/2} p(\{\boldsymbol{0}\} \mid S_{-j})$.\footnote{This could be established formally by induction on the queries of the SVT, using the bound on the $\ell_1$ sensitivity of the counts.}

Since by \Cref{prop:stabind} there are at most two subspaces $U$ for which $n(U)=0$ and $n_{-j}(U)=1$ or vice versa, we conclude, by a union bound, that for any event $E$ (which is just a union of subspaces), 
$p(E \mid S)\DP{\eps,\delta} p(E \mid S_{-j})$. This finishes the first part of the proof.

\smallskip

We now prove that Items 1 and 2 in the theorem statement are satisfied with high probability.
First, observe that the output space $U$ is either $U=\{\boldsymbol{0}\}$ or it is spanned by some of the input vectors. Consequently, Item 1 is satisfied with probability 1. Regarding Item 2, we can assume, without loss of generality, that 
the number of vectors
$\ell=\Omega\left(\frac{d^2}{\eps}\log\left(\frac{d}{\delta}\right)\right)$; if this is not the case, Item 2 is trivially satisfied. Specifically, we assume \(\ell \geq \frac{100d^2}{\eps}\ln\frac{100d}{\delta}\).
For each $i\leq d$ let $\hat m(i) = m(i) + \mathtt{Lap}(4/\eps)$ be the noisy subspace count as defined in Line 5 of \Cref{alg:linspan}. By \Cref{eq:conlap} and a union bound we have that with probability at least $1-\delta/2$:
\begin{equation}\label{eq:mhatm}
(\forall i\leq d): \bigl\lvert m(i) - \hat m(i) \bigr\rvert \leq \frac{4}{\eps}\ln\frac{2d}{\delta}.    
\end{equation}
Another application of \Cref{eq:conlap} yields that with probability at least $1-\delta/2$
\begin{equation}\label{eq:4}
\Bigr\rvert\theta - \frac{16}{\eps}\ln\frac{100d}{\delta}\Bigl\lvert\leq  \frac{2}{\eps}\ln\frac{2}{\delta},%
\end{equation}
where $\theta$ is the cutoff defined in Line 4 of \Cref{alg:linspan}.
Thus, with probability at least $1-\delta$ the inequalities in \Cref{eq:mhatm,eq:4} are satisfied.
It suffices to show that if these inequalities are satisfied then Item 2 is also satisfied: Since (i) $\ell \geq \frac{100d^2}{\eps}\ln\frac{100d}{\delta}$, (ii) there are $d$ possible dimensions of subspaces, and (iii) each set in the partition contains at most $d$ independent vectors, then there exists
$i$ such that $m(i)\geq \frac{\ell}{d^2} \geq \frac{100}{\eps}\ln\frac{100d}{\delta}$. Hence in this event we have,
\[\hat m(i) \geq \frac{100}{\eps}\ln\frac{100d}{\delta} -  \frac{4}{\eps}\ln\frac{2d}{\delta} \geq \frac{20}{\eps}\ln\frac{100d}{\delta} \geq \theta, \]
and therefore the output subspace $U$ satisfies $\mathtt{dim}(U) = j$ for some $j\geq i$.  
Moreover, \Cref{eq:mhatm,eq:4} imply that if $m(j)=0$ then $\hat m(j) <\theta$, which means that such $j$'s are not outputted. In other words the output space $U \not= \{\bf 0\}$. 
Thus, the only input vectors that are possibly outside $U$ are those that belong to an independent set spanning subspaces of
dimension in 
$[j+1,d]$.
 Since $\hat m(\ell) < \theta$ for these subspaces, we get that the number of such vectors is at most
\[d^2\cdot\Bigl( \frac{20}{\eps}\ln\frac{100}{\delta} + \frac{4}{\eps}\ln\frac{2d}{\delta}\Bigr) = O\Bigl(\frac{d^2}{\eps}\log\frac{d}{\delta}\Bigr).\]
\end{proof}

\subsubsection{Synthetic Affine Spans}
The following algorithm uses \Cref{alg:linspan} to privately approximate affine spans. It will serve as a key component of our differentially private (DP) linear programming solver.

\begin{algorithm}
\caption{Synthetic Affine Span}
\label{alg:affine_hulls}
\textbf{Input: }a sequence \(\{\boldsymbol{u_i}\}_{i=1}^m\subseteq \F^d\).\\
\textbf{Output:} a sequence \(\{\boldsymbol{v_j}\}_{j=1}^k\subseteq \F^d\)
     such that \(\mathtt{Affspan}(\{\boldsymbol{v_j}\})\) approximates \(\mathtt{Affspan}(\{\boldsymbol{u_i}\})\) (see \Cref{t:privateaffspan}).
\begin{algorithmic}[1]
\State Apply \Cref{alg:linspan} to the sequence $\{\boldsymbol{u_i'}\}_{i=1}^m$, where $\boldsymbol{u_i'}=(\boldsymbol{u_i}, 1)\in \F^{d}\times \F$. 
\State If \Cref{alg:linspan} outputs \(\{\boldsymbol{0}\}\) then output \(\emptyset\).
\State Else, let $\{\boldsymbol{\boldsymbol{v_i'}}\}_{i=1}^k$ be the sequence of vectors outputted by \Cref{alg:linspan}.
\State Set $\boldsymbol{v_i''} = \frac{1}{\boldsymbol{v_i'}(d+1)}\cdot \boldsymbol{v_i'}$, where \(\boldsymbol{v_i'}=(\boldsymbol{v_i'}(1),\ldots, \boldsymbol{v_i'}(d+1))\).
\State \textbf{Output} the sequence \(\left\{\boldsymbol{v_j}\right\}_{j=1}^k\), where $\boldsymbol{v_j''}=(\boldsymbol{v_j},1)\in\F^d\times \F$.
\end{algorithmic}
\end{algorithm}

\begin{customthm}{B}[Private Affine Span]\label{t:privateaffspan}
\Cref{alg:affine_hulls} is $(\eps,\delta)$-DP and satisfies the following with probability at least~$1-\delta$.
Let $V_{input}=\mathtt{Affspan}(\{\boldsymbol{v_j}\})$ denote the affine span of the input sequence and let $V_{output}=\mathtt{Affspan}(\{\boldsymbol{u_i}\})$
denote the affine span of the output sequence. (In the special case when the output is $\emptyset$, let $V_{output}=\emptyset$ as well.) Then,
\begin{enumerate}
\item $V_{output}\subseteq V_{input}$.
\item All but $O(\frac{d^2}{\eps}\log(d/\delta))$ of the input vectors $\boldsymbol{u_i}$ satisfy $\boldsymbol{u_i}\in V_{output}$.
\end{enumerate}
\end{customthm}
\begin{proof}
\Cref{alg:affine_hulls} inherits the \((\eps,\delta)\)-privacy guarantees of \Cref{alg:linspan},
as follows by the post-processing property of differentially private (DP) algorithms.

We now prove the utility guarantees. We rely on the following observation.
\begin{obs}
Let $\boldsymbol{Z_i'}\in\mathbb{F}^{d+1}$ for $i=0,\ldots, n$ be vectors of the form
\(\boldsymbol{Z_i'} = (\boldsymbol{z_i},1)\), where \(\boldsymbol{z_i}\in \F^d\).
Then, the following statements are equivalent:
\begin{enumerate}\label{obs:linaff}
    \item \(\boldsymbol{z_0'}\) is a linear combination of the \(\boldsymbol{z_i'}\)'s, i.e.,
    \(\boldsymbol{z_0'}\in\mathtt{span}(\{\boldsymbol{z_i'}\})\).%
    \item \(\boldsymbol{z_0'}\) is an affine combination of the \(\boldsymbol{z_i'}\)'s, i.e.,
    \(\boldsymbol{z_0'}\in\mathtt{Affspan}(\{\boldsymbol{z_i'}\})\).
    \item \(\boldsymbol{z_0}\) is an affine combination of the \(\boldsymbol{z_i}\)'s, i.e.,
    \(\boldsymbol{z_0}\in\mathtt{Affspan}(\{\boldsymbol{z_i}\})\).
\end{enumerate}
\end{obs}
By \Cref{prop:linspan} we have that with probability at least~$1-\delta$, $\mathtt{span}(\{\boldsymbol{v_i'}\})\subseteq \mathtt{span}(\{\boldsymbol{u_i'}\})$, and \(\mathtt{span}(\{\boldsymbol{v_i'}\})\) contains all but $O(\frac{d^2}{\eps}\log(d/\delta))$ of the $\boldsymbol{u_i'}$'s.
In the remainder of the proof we assume this event holds. 
First, we show that the $\boldsymbol{v_j''}$'s are well defined. Indeed, every nonzero vector $\boldsymbol{v}\in \mathtt{span}(\{\boldsymbol{u_i'}\})$ satisfies $\boldsymbol{v}(d+1)\neq 0$, and therefore $\boldsymbol{v_i'}(d+1)\neq 0$ for all $\boldsymbol{v_i'}$.

We now prove Item 1: let $\boldsymbol{v}\in V_{output}$; thus, $\boldsymbol{v}$
is an affine combination of the $\boldsymbol{v}_i$'s.
Hence, by~\Cref{obs:linaff} $(\boldsymbol{v},1)$ is a linear combination of the~$\boldsymbol{v_i'}$'s.
Since $\mathtt{span}(\{\boldsymbol{v_i'}\})\subseteq \mathtt{span}(\{\boldsymbol{u_i'}\})$
it follows that $(\boldsymbol{v},1)$ is also a linear combination of the \(\boldsymbol{u_i'}\)'s;
finally, since $\boldsymbol{u}'_i=(\boldsymbol{u_i},1)$ it follows that $\boldsymbol{v}$ is an affine combination of the \(\boldsymbol{u_i}\)'s and hence $\boldsymbol{v}\in V_{input}$ as required.

For Item 2, let $\boldsymbol{u_i}$ such that \(\boldsymbol{u}'_i\in \mathtt{span}(\{\boldsymbol{v_i'}\})\)
(recall that this is satisfied by all but $O(\frac{d^2}{\eps}\log(d/\delta))$ of the $\boldsymbol{u_i}$'s). Thus, $\boldsymbol{u_i'}\in \mathtt{span}(\{\boldsymbol{v_j''}\})$,
and since $\boldsymbol{u_i'}=(\boldsymbol{u_i},1)$, $\boldsymbol{v_j''}=(\boldsymbol{v_j},1)$,
it follows that \(\boldsymbol{u_i}\) is an affine combination of the \(\boldsymbol{v_j}\)'s as required.

\end{proof}

\subsubsection{Synthetic Systems of Linear Equations}

\begin{algorithm}
\caption{Synthetic Systems of Linear Equations}
\label{alg:lin_equation}
\textbf{Input:} linear system \(\mathcal{I}=\left\{\boldsymbol{a_i}\cdot x = b_i\right\}_{i=1}^m\).\\
\textbf{Output:} linear system \(\mathcal{O}=\{\boldsymbol{c_j}\cdot x = d_j \}_{j=1}^k\) that approximates $\I$ (see \Cref{t:privatelineq})
\begin{algorithmic}[1]
\State Apply \Cref{alg:linspan} on the sequence $\{\boldsymbol{v}_i\}_{i=1}^m$, where $\boldsymbol{v}_i=(\boldsymbol{a}_i, -b_i)\in \F^{d}\times \F$.
\State Let $\{\boldsymbol{u_i}\}_{i=1}^k$ be the sequence of vectors that \Cref{alg:linspan} outputs.
\State \textbf{Output} the linear system \(\left\{\boldsymbol{c_j}\cdot \boldsymbol{x} = d_j\right\}_{i=1}^k\), where $\boldsymbol{u_j}=(\boldsymbol{c_j},-d_j)\in\F^d\times \F$.
\end{algorithmic}
\end{algorithm}

\begin{customthm}{C}[Private Synthetic Linear Equations]\label{t:privatelineq}
\Cref{alg:lin_equation} is $(\eps,\delta)$-DP and satisfies the following with probability at least $1-\delta$.
Let $\mathtt{sol(\I)}$ be the solution space to the input system and let $\mathtt{sol(\O)}$ be the solution space to the output system.
Then,
\begin{enumerate}
\item $\mathtt{sol(\I)}\subseteq\mathtt{sol(\O)}$.
\item If $\mathtt{sol(\O)}\neq\emptyset$ then each $x\in \mathtt{sol(\O)}$ satisfies all but $O(\frac{d^2}{\eps}\log(d/\delta))$ of the equations in $\mathcal{I}$.
\end{enumerate}
\end{customthm}
\begin{proof}
    Due to the post-processing property of differentially private (DP) algorithms, \Cref{alg:lin_equation} inherits the \((\eps,\delta)\)-privacy guarantees of \Cref{alg:linspan}.
We now prove the utility guarantees. We use the following basic fact: 
\begin{obs}\label{f:aff-lin}
Let $\boldsymbol{w_0},\ldots \boldsymbol{w_k}$
be $d+1$-dimensional vectors each of the form $\boldsymbol{w_i} = (\boldsymbol{r_i},-t_i)$, where $\boldsymbol{r_i}$ is a $d$-dimensional vector and $t_i$ is a scalar. Then, the following statements are equivalent:
\begin{enumerate}
    \item $\boldsymbol{w_0}\in\mathtt{span}(\{\boldsymbol{w_i}: i>0\})$.
    \item Any solution $\boldsymbol{x}\in\mathbb{F}^d$
satisfying $\boldsymbol{x}\cdot \boldsymbol{r_i}=t_i$ for all $i>0$, also satisfies $\boldsymbol{x}\cdot \boldsymbol{r_0}=t_0$.
\end{enumerate} 
\end{obs}
By \Cref{prop:linspan} we have that with probability at least~$1-\delta$, $\mathtt{span}(\{\boldsymbol{u_i}\})\subseteq \mathtt{span}(\{\boldsymbol{v_i}\})$, and \(\mathtt{span}(\{\boldsymbol{u_i}\})\) contains all but $O(\frac{d^2}{\eps}\log(d/\delta))$ of the $\boldsymbol{v_i}$'s.
Assume this event holds. 
\Cref{f:aff-lin} therefore yields that $\mathtt{sol(\I)}\subseteq\mathtt{sol(\O)}$ as stated in Item 1. 
For Item 2, let \(\boldsymbol{x}\in\mathtt{sol(\O)}\); \Cref{f:aff-lin} implies that $\boldsymbol{x}$ 
satisfies each input equation $\boldsymbol{x}\cdot \boldsymbol{a_i}=b_i$ such that $\boldsymbol{v_i}\in \mathtt{span}(\{\boldsymbol{u_j}\})$.
Item 2 therefore follows, because this is satisfied by all but at at most $O(\frac{d^2}{\eps}\log(d/\delta))$ of the $\boldsymbol{v_i}$'s.
\end{proof}

\section{Polyhedra and Linear Inequalities}\label{sec:inequalitiesmain}

 \cite{DS08} gave a perceptron-based algorithm to solve linear programming. 
They consider LP in the form: Given a matrix $A$, find a nonzero vector $x\in \mathcal{R}^d$, such that $Ax\ge 0$. (The standard LP formulation can be reduced to this setting.) 
They define 
the {\em roundness} of the feasible region of $A$ to be 
$$
\rho(A)=\max_{x:\|x\|_2=1, Ax\geq0} \min_{a_i}  \langle \bar{a_i},x \rangle,
$$
where $\bar{a_i}$ is the unit vector along $a_i$.
At a high level, in every iteration they transform the linear system to a system with 
a larger $\rho$. When 
the roundness becomes $\Omega(1/d)$ then in $O(d^2)$ additional steps the standard perceptron algorithm find the required feasible vector.
Their algorithm is polynomial in the size of $A$ and $\log(1/\rho)$.\footnote{Their algorithm takes square roots in order to normalize vectors. \cite{DS08} do not address the question of how to implement the algorithm using finite precision so that it still runs in polynomial time.}

In Appendix \ref{appendix:vempala} we work out the details of a differentially private version of this algorithm (which we call the DP-DV-algorithm).
The crucial idea that allows one to do this is the fact that we can make the perceptron steps (and perceptron-like steps that are used to improve the roundness) by averaging many ``violated constraints'' rather than using a single one. Then by carefully privatizing these averages we can make the algorithm differentially private. The resulting algorithm satisfies all the constraints except for a polynomially (in $d$) many and is computationally efficient.

This result is summarized in the following theorem which we prove in Appendix \ref{appendix:vempala}.

\begin{theorem} \label{thm:vempala}
Denote $T=\Theta(d\cdot\ln(1/\rho_0)+\ln(1/\beta))$. 
There exists a computationally efficient $(\epsilon,\delta)$-differentially private algorithm whose input is an $n\times d$ matrix $A$ and its output is a vector $x^*\in\R^d$ such that the following holds. Let $a_1,\ldots,a_n$ denote the rows of $A$. If $\rho(A)\geq\rho_0$, then with probability at least $1-\beta$ the returned vector $x^*$ satisfies 
\begin{eqnarray*}
\left|\left\{i: \langle \bar{a}_i,\bar{x}^* \rangle<\frac{2^{-T}}{12000}\right\}\right|&\leq& \Gamma = O\left(
\frac{d^5}{\epsilon}
 \ln^{1.5}(d) \ln^2\left(\frac{d\ln(\frac{1}{\rho_0})}{\beta}\right)  \ln^{1.5}\left(\frac{1}{\rho_0}\right) \ln\left(\frac{d\ln(\frac{1}{\rho_0})}{\beta\delta}\right)
 \sqrt{\ln\frac{1}{\delta}}
 \right)\\
 &=&\frac{d^5}{\epsilon}\cdot{\rm polylog}\left(d,\frac{1}{\beta},\frac{1}{\delta},\frac{1}{\rho_0}\right).
\end{eqnarray*}
\end{theorem}

\begin{remark}
Since $2^{-T}/12000$ is positive, it follows, in 
 particular that with probability at least $1-\beta$ the DP-DV-algorithm returns a non-zero vector $x^*$ satisfying
$$
\left|\left\{i: \langle \bar{a}_i,\bar{x}^* \rangle<0\right\}\right|
 \leq\frac{d^5}{\epsilon}\cdot{\rm polylog}\left(d,\frac{1}{\beta},\frac{1}{\delta},\frac{1}{\rho_0}\right).
$$
\end{remark}

\cite{DS08} (Section 4) show how to transform an LP in standard form (say, find $x$ s.t.\ $Ax\le b$, $x\ge 0$) to their homogenous form (find $x\not= 0$ s.t.\ $Ax\ge 0$). This transformation combines the non-negativity constraints $x >0$ with an homogenized version of the $Ax\le b$ constraints. They do this by extending the matrix $A$ with an additional column corresponding to $b$ (and a new variable $x_0$) and additional rows which form the identity matrix. 
In a differentially private setting we want to be private with respect to the constraints 
$Ax\le b$, but the constraints 
$x \ge 0$ are public.
This gives rise to an instance of the problem: find $x\not= 0$ such that $Ax\ge 0$, in which we want to be private with respect to some of the constraints, but other constraints are public.

To support this setting we extend 
the DP-DV-algorithm to get as input an additional set of constraints of the form $Bx\ge 0$ which are not private and have to be satisfied.
An easy way to do this is by duplicating each of these constraints more times than the upper bound on the number of constraints that are allowed to be violated, as specified in \Cref{thm:vempala}, and  treating these duplicates as any of the other constraints. We can also modify the algorithm such that it treats these constraints specially (we omit the details).
Combining this with the
 reduction in Section 4 of \cite{DS08} we
 obtain the following theorem.

\begin{theorem}\label{thm:dp-lp1}
Let \(\mathcal{LP}=\left\{Ax\le b, x\ge 0\right\}\) be 
a system of linear inequalities, where $A$ is an $m\times d$ matrix, $b\in \mathcal{R}^d$, and the feasible region of $A$ contains a ball of radius $\rho_0$.
Let 
 $\epsilon, \delta$ be privacy parameters and $\beta$ a confidence parameter.
There exists a computationally efficient $(\epsilon,\delta)$-differentially private algorithm that with probability $1-\beta$
returns a vector $x^*\geq0$ that satisfies 
 all but
$\frac{d^5}{\epsilon}{\rm polylog}(d,1/\delta,1/\beta,1/\rho_0)$ of the constraints $Ax\le b$.
\end{theorem}

\subsection{When the feasible region is not fully dimensional} \label{sec:no-margin}

The Algorithm of \cite{DS08}
cannot solve LP's with roundness $0$.\footnote{Note that an LP can have roundness $0$ even when the coefficients of the constraints lie on a discrete grid. This happens whenever the feasible region is a flat of dimension smaller than $d$.} 
A standard way to overcome this hurdle (that arises also when using the non-private Ellipsoid algorithm) is to slightly perturb the linear program such that the perturbed version is feasible iff the original version is feasible, and furthermore the feasible region of the perturbed LP has positive volume.
The following lemma specifies this perturbation.
 
\begin{lem} \label{lem:Peps}
The linear program
\begin{eqnarray*}
 Ax & \le & b  \;\;\; (P) \\
 x & \ge & 0  
\end{eqnarray*}
is feasible if and only if the linear program 
\begin{eqnarray*}
 Ax & \le & b +\eta {\bf 1} \;\;\; (P^\eta) \\
 x & \ge & 0  
\end{eqnarray*}
is feasible.
Here $A$ is an $m \times d$ integer matrix, 
$b$ is an integer vector of length $m$, and 
$\eta=\frac{1}{2(d+1)((d+1)U)^{(d+1)}}$, where 
$U$ is an upper bound on the absolute value of $b_i$ and $A_{ij}$ for all $i\in [m]$ and  $j\in [d]$. We assume that $m\gg d$.
\end{lem}
\begin{proof}
Clearly if $(P)$ is feasible then so is 
$(P^\eta)$.
For the converse assume that $(P)$ is infeasible.
Add to (P) an artificial objective 
\[\max \;\; 0\cdot x  \]
The dual to $(P)$ is
\begin{eqnarray*}
 \lefteqn{ \hspace*{-0.2in} \min \;  b^T y}  \\
 A^T y & \ge & 0  \;\;\;  (D) \\
 y & \ge & 0  
\end{eqnarray*}

Clearly $(D)$ is feasible ($y=0$ is a solution), and therefore since we assume that $(P)$ is infeasible, then by LP duality $(D)$ must be unbounded.
It follows (by convexity) that the polyhedron defined by
\begin{eqnarray*}
  b^T y & = & -1  \\
 A^T y & \ge & 0  \;\;\;  (D')  \\
 y & \ge & 0  
\end{eqnarray*}
is not empty.

By Lemma \ref{lem:cramer} (in Appendix \ref{app:fulld}), there is a point $z$ on the boundary of the feasible region of $D'$, s.t.\ (1) at most $d+1$ of the coordinates of $z$ are non-zero and (2) $|z_i|\le ((d+1)U)^{d+1}$ for all $i\in [m]$. So from our definition of $\eta$ follows that
\[
(b + \eta {\bf 1})^T z =
b^T z + 1/2 \le -1/2 \ .
\]
This implies that
\[
(b + \eta {\bf 1})^T (c\cdot z) =
c(b^T z + 1/2) \le -c/2 \ .
\]
 for any constant $c\ge 0$.
So we conclude that the dual 
$D^\eta$ of $P^\eta$ (with a zero objective added)
given by
\begin{eqnarray*}
 \lefteqn{ \hspace*{-0.2in} \min \;  (b+\eta {\bf 1})^T y}  \\
 A^T y & \ge & 0  \;\;\;  (D^\eta) \\
 y & \ge & 0  
\end{eqnarray*}
is also unbounded.
Thus $P^\eta$ is infeasible. 
\end{proof}

It is not hard to show that if
$P$ is feasible then the roundness of $P^\eta$ is at least $\eta/(d\cdot U)$.\footnote{If $P$ is feasible then $P^\eta$ contains a cube of side length  $\eta/dU$.}
It follows from Lemma~\ref{lem:Peps},  that if we are only interested to privately decide if the system of linear constraints 
\begin{eqnarray*}
 Ax & \le & b  \;\;\; (P) \\
 x & \ge & 0  
\end{eqnarray*}
is feasible then we can apply the DP-DV-algorithm
to the system $(P^\eta)$.
It is still not clear, however, how to obtain a feasible point to $(P)$ as in Theorem \ref{thm:vempala}. We show how to obtain such a feasible point next. 

By applying the DP-DV-algorithm to $P^\eta$ we obtain a point $x^{*}$ that satisfies most of the constraints of $P^\eta$. Let 
$\JJJ$ be indices of the constraints  of $P^\eta$ that $x^{*}$ satisfies, and let $A_{\JJJ}x \le b$ be the (original) unperturbed version of these constraints. We can partition $\JJJ$ into two subsets. One, say $\JJJ_1$, consists of the indices of the constraints that are satisfied by $x^{*}$, and the other, say $\JJJ_2 = \JJJ\setminus \JJJ_1$, consists of the indices of the constraints that are not satisfied by $x^*$,
that is $b_{\JJJ_2}\le A_{\JJJ_2}x^* \le b_{\JJJ_2} + \eta {\bf 1}_{\JJJ_2}$.

Lemma \ref{lem:Peta2} in Appendix \ref{app:fulld} proves that the system in which we turn all the inequalities in $A_{\JJJ}x \le b$  of an index in $\JJJ_2$ into equalities must also be feasible. Its proof  is very similar to the proof of Lemma \ref{lem:Peps}.
Based on Lemma \ref{lem:Peps} and Lemma \ref{lem:Peta2} we  now describe the following algorithm for DP solving general linear programs.
The idea is to apply the DP-DV-algorithm to the perturbed system. Then identify the set of constraints that we can turn to equalities, privatize these equalities using \Cref{alg:lin_equation},  use their sanitized version to reduce the dimension of the linear program, and iterate with the new LP.

\begin{algorithm}[h!]
\caption{DP-LP \label{alg:dp-lp}}
\textbf{Input:} A system of linear inequalities \(\mathcal{LP}=\left\{Ax\le b, x\ge 0\right\}\), $A$ is $m\times d$, $b\in \mathcal{R}^d$, all entries of $A$ and $b$ are integers bounded by $U$ in absolute value. Privacy parameters $\epsilon, \delta$, confidence parameter $\beta$.\\
\textbf{Output:} A feasible solution $x^*$, s.t.\ $x^*_i\ge 0$, $i\in [d]$ and $x^*$ satisfies all but
$\frac{{\rm poly}(d)}{\eps}{\rm polylog}(1/\delta,1/\beta,U)$ of the constraints defined by $A$ and $b$ of $\mathcal{LP}$.
\begin{algorithmic}[1]
\State
Define the roundness parameter
$\rho=\eta/(d\cdot U)$, where 
$\eta=\frac{1}{2(d+1)((d+1)U)^{(d+1)}}$.

\State Apply the DP-DV-algorithm to 
the linear program 
$\mathcal{LP}^\eta$ defined  in \Cref{lem:Peps},
with roundness parameter
$\rho$.

\noindent
 Let $x^*$ be the
solution that the DP-DV-algorithm returns, and let $\mathcal{J}$ be the indices of the constraints among 
$Ax\le b+ \eta{\bf 1}$ that $x^*$ satisfies.
Let $\mathcal{J}_1$ be the subset of the constraints among
$A_{\mathcal{J}}x \le b_{\mathcal{J}}$ that $x^*$ 
satisfies and let 
 $\mathcal{J}_2=\JJJ-\JJJ_1$.
\State
If $|\mathcal{J}_2| =O(\frac{d^2}{\eps}\log(d/\delta))$ (with the appropriate constant as implied from the proof of \Cref{prop:linspan}.) then return $x^*$.
\State Apply \Cref{alg:lin_equation}
to sanitize the system of equations 
$A_{\mathcal{J}_2}x = b_{\mathcal{J}_2}$. Let
$Cx=d$ be the sanitized system.
Denote by $t$ the number of independent rows in $C$ (Note that we must have $t\le d$).
If $t=d$ or $|\JJJ_1| < \Gamma$ (where $\Gamma$ is as defined in Theorem \ref{thm:vempala} with $\rho=\eta/(d\cdot U)$) then return a solution to $Cx=d$.
\State  
Consider the system 
$\{A_{\JJJ_1}x \le b_{\JJJ_1}, x\ge 0 \}$. Reduce this system to dimension $d-t$ using the $t$ independent rows in $Cx=d$.
Replace $\mathcal{LP}$ with this new system, Replace $U$ by $U^t (\le U^d)$ and go back to Step 1. \newline
\quad\gray{\% Note that when we isolate a variable from an equation of $Cx=d$, and eliminate it from the inequalities  $A_{\JJJ_1}x \le b_{\JJJ_1}, x\ge 0$ then the coefficients in the resulting system may not be integers  anymore. To make the substitution and keep the coefficients integers, we have to multiply the inequalities by the coefficient of the variable which we isolated. Since we isolate $t$ variables, this blows up the upper bound $U$ on the absolute value of the coefficients. Clearly however $U^t$ is a valid new upper bound.}
\end{algorithmic}
\end{algorithm}

\newpage

The following theorem follows using standard composition theorems for DP-algorithms.

\begin{customthm}{D}[Differantially-Private Linear Programming]\label{thm:dp-lp}
Let \(\mathcal{LP}=\left\{Ax\le b, x\ge 0\right\}\) be 
a feasible system of linear inequalities, where $A$ is an $m\times d$ matrix, $b\in \mathcal{R}^d$, and all entries of $A$ and $b$ are integers bounded by $U$ in absolute value.
Let 
 $\epsilon, \delta$ be privacy parameters and $\beta$ a confidence parameter. Algorithm \ref{alg:dp-lp}
is a computationally efficient $(\epsilon,\delta)$-differentially private such that with probability $1-\beta$
returns a vector $x^*\geq0$ that satisfies 
 all but
$\frac{{\rm poly}(d)}{\eps}{\rm polylog}(1/\delta,1/\beta,U)$ of the constraints $Ax\le b$.

\end{customthm}

\subsection{Point in the Convex Hull via Linear Programming and Affine Spans}

In this section, we present our private algorithm for finding a point in the convex hull of a given input set $S\subseteq \XXX^d$, where $\XXX^d\subseteq\R^d$ is a finite $d$-dimensional grid. Intuitively, the algorithm mimics the classical ellipsoid algorithm: We start with a large ellipsoid $\EEE$ that bounds our (finite) input domain and iteratively refine this ellipsoid to approach the convex hull of $S$. Specifically, in each iteration, we formulate an LP whose solution (if it exists) is a hyperplane that separates the center of our current ellipsoid, $c$, from $S$. This LP is feasible if and only if the center $c$ lies outside the convex hull of $S$. We then execute the private LP algorithm from Theorem~\ref{thm:vempala} to attempt to solve this LP. If it fails, then we know that $c$ is a valid solution. Otherwise, we obtain a separating hyperplane, which we use to ``cut'' the ellipsoid. This iterative process reduces the ellipsoid's volume. We prove that if the volume of $\Conv(S)$ remains positive,\footnote{For technical reasons, throughout the execution we sometimes delete points from $S$. Thus, it is possible that the volume of $\Conv(S)$ is initially positive but drops to zero at some point during the execution.} the center of the ellipsoid must enter $\Conv(S)$ after at most $T\approx d^2$ iterations. Therefore, after $T$ iterations, we either find a valid solution or discover that $\Conv(S)$ has zero volume, in which case we restart the algorithm in a lower-dimensional subspace. To this end, we use our techniques from Section~\ref{sec:equalitiesmain} to privately identify such a suitable subspace.

A technical issue that we would like to avoid here is that in order to ``restart the algorithm in a lower-dimensional subspace'', our algorithm must be able to operate not only in $\R^d$, but also in lower-dimensional affine subspaces of it, which would complicate the algorithm. To avoid this, we begin by presenting a (non-private) algorithm, \texttt{Aff2Lin}, for computing a convexity-preserving transformation from a $k$-dimensional affine subspace to $\R^k$. This allows us to remain within a Euclidean space even after discovering that $\Conv(S)$ has zero volume, thus simplifying the algorithm.

\begin{algorithm}[H]
\caption{\texttt{Aff2Lin}}\label{alg:aff2Lin}

{\bf Input:} Affinely independent vectors $u_1,\dots,u_k\in\R^q$ for $k\leq q$.
\begin{enumerate}[rightmargin=10pt,itemsep=1pt]
\item For $i\in[k]$ denote $u'_i=(u_i,1)\in\R^{q+1}$.

\item Let $M$ be the $k\times (q+1)$ matrix whose rows are $u'_1,\dots,u'_k$.

\item Let $c_1,\dots,c_k\in[q+1]$ be the indices of $k$ independent columns in $M$ such that $c_k=q+1$, and let $A$ be the $k\times k$ sub-matrix of $M$ containing these columns. \quad\gray{\% A simple calculation shows that $\{u_i\}$ are affinely independent iff $\{u'_i\}$ are linearly independent. Hence, there are $k$ independent columns in $M$ and this step is well-defined.}

\item Let $T=A^{-1}\,M$, and let $t_1,\dots,t_k$ denote the rows of $T$. \quad\gray{\% $T$ is a $k{\times}(q+1)$ matrix containing the $k{\times}k$ identity matrix as a sub-matrix in columns $c_1{,}...{,}c_k$.}

\item Return the following two functions:
\begin{itemize}
    \item ${\rm Project}:\R^q\rightarrow\R^{k-1}$. This function takes a vector $v\in\R^q$ and returns the vector $(v[c_1],v[c_2],\dots,v[c_{k-1}])\in\R^{k-1}$.

    \item ${\rm GoUp}:\R^{k-1}\rightarrow\R^q$. This function takes a vector $v\in\R^{k-1}$, computes $v''=v[1]t_1+v[2]t_2+\dots+v[k-1]t_{k-1}+1\cdot t_k\in\R^{q+1}$, and returns the vector $v'''\in\R^q$ obtained from $v''$ by omitting its last coordinate.
\end{itemize}
    
\end{enumerate}
\end{algorithm}

\begin{remark}
We stress that \texttt{Aff2Lin} preserves the dimension of the {\em affine span} of the input vectors $u_1,\dots,u_k$, rather than the dimension of their (linear) span. In particular, when $k=1$, then $\Affspan(u_1)$ has dimension 0 while $\span(u_1)$ has dimension 1. In this case, running $\texttt{Aff2Lin}(u_1)$ results in a projection onto $\R^0$, which is a one-point space that contains only the empty tuple.
\end{remark}

\begin{lem}\label{lem:aff2lin}
$S=\{x_1,\dots,x_n\}\subseteq\R^q$ be a multiset and let $\{u_i\}_{i=1}^k\subseteq \R^q$ be affinely independent vectors such that  $S\subseteq\Affspan(u_1,\dots,u_k)$. Let $\left( {\rm Project}, {\rm GoUp} \right)\leftarrow\texttt{Aff2Lin}(u_1,\dots,u_k)$, and let $\hat{S}=\{{\rm Project}(x_i)\,:\,i\in [n]\}\subseteq\R^{k-1}$. Then, for every point $v\in\Conv(\hat{S})$ it holds that ${\rm GoUp}(v)\in\Conv(S)$.
\end{lem}

\begin{proof}
As in the statement of the lemma, we denote
$$\hat{S}=\{{\rm Project}(x_i)\;:\;i\in [n]\}=\{(x_i[c_1],\dots,x_i[c_{k-1}])\;:\;i\in [n]\}\subseteq\R^{k-1}.$$
Now let $v\in\R^{k-1}$, and suppose that $v$ is a convex combination of the points in $\hat{S}$, say
$$
v=
\sum_{i\in[n]}\alpha_i\begin{pmatrix}  x_i[c_1] \\ \vdots \\ x_i[c_{k-1}]  \end{pmatrix}
=
\begin{pmatrix}  \sum_i \alpha_i x_i[c_1] \\ \vdots \\ \sum_i \alpha_i x_i[c_{k-1}] \end{pmatrix}
\qquad\qquad\text{\gray{\% For $\alpha_i\geq 0$ satisfying $\sum_i\alpha_i=1$.}}
$$
We need to show that ${\rm GoUp}(v)\in\Conv(S)$. To this end, let us consider the vector $v''\in\R^{q+1}$ defined in the execution of ${\rm GoUp}(v)$:
\begin{eqnarray}
v'' &=& v[1]t_1+v[2]t_2+\dots+v[k-1]t_{k-1}+1\cdot t_k\nonumber\\
&=&\left(\sum_{i\in[n]}\alpha_i\,x_i[c_1]\right)t_1+\dots+\left(\sum_{i\in[n]}\alpha_i\,x_i[c_{k-1}]\right)t_{k-1} + 1\cdot t_k\nonumber\\
&=&
\left(\sum_{i\in[n]}\alpha_i\,x_i[c_1]\right)t_1+\dots+\left(\sum_{i\in[n]}\alpha_i\,x_i[c_{k-1}]\right)t_{k-1} + \left(\sum_{i\in[n]}\alpha_i\right)\cdot t_k\nonumber\\
&=&
\sum_{i\in[n]}
\alpha_i \Big( x_i[c_1]t_1+\dots+x_i[c_{k-1}]t_{k-1}+t_k  \Big)\label{eq:GoUp1}.
%
\end{eqnarray}
We will show that for every vector $x\in\Affspan(u_1,\dots,u_k)$ it holds that
\begin{eqnarray}\label{eq:GoUp2}
(x,1)=x[c_1]t_1+x[c_2]t_2+\dots+x[c_{k-1}]t_{k-1}+1\cdot t_k.    
\end{eqnarray}
Thus, continuing from Equality~(\ref{eq:GoUp1}), we have that
$$
v''=\sum_{i\in[n]}\alpha_i \cdot (x_i,1),
$$
This shows that $v''$ is a convex combination of $\{(x_i,1)\}_{i\in[n]}$. Thus ${\rm GoUp}(v)$, which is obtained from $v''$ by omitting its last coordinate, must be a convex combination of $S=\{x_i\}_{i\in[n]}$.

\medskip

It remains to prove Equality~(\ref{eq:GoUp2}). 
To this end, denote $U=\span(u'_1,\dots,u'_k)=\span\big((u_1,1),\dots,(u_k,1)\big)$, and observe that, as in Remark~\ref{rem:canonB}, the rows of the matrix $T$, denoted as $t_1,\dots,t_k$, are a basis for $U$. Thus, every vector in $y\in U$ has a unique representation as a linear combination of $t_1,\dots,t_k$. Recall that the matrix $T$ contains the $k{\times}k$ identify matrix as a sub-matrix in columns $c_1,\dots,c_k$. Therefore, for every vector $y\in U$, the coefficients of the linear combination of $t_1,\dots,t_k$ that produce $y$ are exactly $y[c_1],y[c_2],\dots,y[c_k]$. That is, for every $y\in U$ we have
\begin{eqnarray}\label{eq:GoUp}
y=y[c_1]t_1+y[c_2]t_2+\dots+y[c_k]t_k.    
\end{eqnarray}
Thus, to prove Equality~(\ref{eq:GoUp2}), it suffices to show that for every $x\in\Affspan(u_1,\dots,u_k)$ it holds that $y=(x,1)\in\span(u'_1,\dots,u'_k)=U$. This is straightforward. Indeed, as $x\in\Affspan(u_1,\dots,u_k)$, for appropriate coefficients $\gamma_1\dots,\gamma_k$ (whose sum is 1) we have $x=\sum_{i\in [k]}\gamma_i u_i$, and thus
$$
\sum_{i\in [k]}\gamma_i u'_i
=
\sum_{i\in [k]}\gamma_i\begin{pmatrix}  u_i \vspace{-5pt}\\ \color{gray} \rule{0.05cm}{0.05mm} \,\rule{0.05cm}{0.05mm} \,\rule{0.05cm}{0.05mm} \\ 1 \end{pmatrix}
=
\begin{pmatrix}  \sum\gamma_i u_i \vspace{-3pt}\\\vspace{2pt} \color{gray} \rule{0.05cm}{0.05mm} \,\rule{0.05cm}{0.05mm} \,\rule{0.05cm}{0.05mm}\,\rule{0.05cm}{0.05mm}\,\rule{0.05cm}{0.05mm}\,\rule{0.05cm}{0.05mm}\,\rule{0.05cm}{0.05mm} \\ \sum\gamma_i \end{pmatrix}
=
\begin{pmatrix}  x \vspace{-5pt}\\ \color{gray} \rule{0.05cm}{0.05mm} \,\rule{0.05cm}{0.05mm} \,\rule{0.05cm}{0.05mm} \\ 1 \end{pmatrix},
$$
showing that $(x,1)\in\span(u'_1,\dots,u'_k)$.
\end{proof}

\begin{definition}
We write $\EEE\left(c,\{v_i\}_{i\in[d]},\{\lambda_i\}_{i\in[d]}\right)$ to denote the ellipsoid in $\R^d$ with center $c\in\R^d$, principal axes $v_1\dots,v_d\in\R^d$ (where these form an orthonormal basis of $\R^d$), and principal radii $\lambda_1,\dots,\lambda_d>0$. That is,
$$\EEE\left(c,\{v_i\}_{i\in[d]},\{\lambda_i\}_{i\in[d]}\right)=\left\{
c + \sum_{i\in[d]} \alpha_i v_i \;:\; (\alpha_1,\dots,\alpha_d)\in\R^d \text{ such that } \sum_{i\in[d]}\left(\frac{\alpha_i}{\lambda_i}\right)^2\leq 1
\right\}.$$ 
\end{definition}

\begin{algorithm}[H]
\caption{\texttt{PinHull}}
{\bf Notations:} Let $X,d\in\N$ be parameters, let $\XXX=\Big\{\frac{a}{X} \;:\; a\in[-X,X]\cap\Z\Big\}$, and let $\XXX^d$ be a $d$-dimensional grid.

\smallskip
{\bf Input:} Multiset $S$ of points in $\XXX^d$ and privacy parameters $\epsilon,\delta$.

\smallskip
{\bf Tools used:} An $(\epsilon,\delta)$-DP algorithm for solving linear programs of the form $Ax\geq0$, that w.p.\ $(1-\beta)$ returns a solution that violates at most $\Gamma$ constraints (see Theorem~\ref{thm:vempala}). 
An $(\epsilon,\delta)$-DP algorithm for approximating affine spans (see Theorem~\ref{t:privateaffspan}). A non-private algorithm, \texttt{Aff2Lin}, for computing a convexity-preserving transformation from a $(k-1)$-dimensional affine subspace to $\R^{k-1}$ (see Lemma~\ref{lem:aff2lin}).  

\begin{enumerate}[rightmargin=10pt,itemsep=1pt]

\item\label{step:refinedGrid} Let $Y=\Theta\left( d^{3.5}  d!\,X^d \, d^{d/2} \, (2T)^{dT} \right)$ be a multiple of $X$, let $\YYY=\Big\{\frac{a}{Y} \;:\; a\in[-Y,Y]\cap\Z\Big\}$, and let $\YYY^d$ be a $d$-dimensional refined grid. Note that $\XXX^d\subseteq\YYY^d$. \qquad\gray{\% Even though all input points come from the grid $\XXX^d$, some steps in our algorithm require points with higher precision, such as those from the refined grid $\YYY^d$.}

\item Denote $q=d$ and let ${\rm GoAllUp}:\R^q\rightarrow\R^d$ denote the identify function. \qquad\gray{\% Throughout the execution, we may discover that the convex hull of our input points is not full-dimensional. When this happens, we restart the algorithm in a subspace of smaller dimension $q<d$ and maintain an appropriate ``reverse'' function, ${\rm GoAllUp}$, which allows us to translate a solution from the lower dimensional subspace back to a solution in the original space.}

\item\label{step:base} If $q=0$ then let $c=()\in\R^0$ be the empty vector 
and return ${\rm GoAllUp}(c)$. \qquad\gray{\% $\R^0$ is a one-point space that contains only the empty tuple. This is the base case of the algorithm.}

\item Let $\EEE$ be the $q$-dimensional ball of diameter $\sqrt{q}$ around the origin of $\XXX^q$. \quad\gray{\% Note that $\EEE$ contains all of $\XXX^q$ and in particular contains $\Conv(S)$.}

\item\label{step:ElipLoop} For at most $T=\Theta(d^2\ln(dX))$ rounds do: 
\qquad \gray{\% See Claim~\ref{claim:T} for the analysis that bounds the number of rounds.}
\begin{enumerate}

    \item\label{step:callDPLP} Let $c$ be the center of the current ellipsoid $\EEE$ and let $A$ be the $|S|\times q$ matrix whose rows are $(z-c)\in\R^q$ for every $z\in S$. Run the private algorithm for LP (see \Cref{thm:vempala}) to find an approximate solution $x^*\neq0$ to the system $A x\geq 0$. Use roundness parameter $\rho=\frac{1}{d \, (d!)^{2d+2} \, (2Y)^{2d^2 + 2d} \, Y^d}$. \newline \gray{\% We run the private LP algorithm to try and separate $c$ (the center of the current ellipsoid) from the points in $S$. By the discrete separation theorem (see Lemma~\ref{lem:Separation}), if $c$ is not in the convex hull of $S$ then this LP is feasible with roundness $\rho$, and the private LP algorithm should succeed in finding a solution $x^*$.}

    \item\label{step:verifyDPLP} Let $\#_{\rm violated}(A,x^*)$ denote the number of constraints in $A$ that $x^*$ violates, i.e., the number of points $z\in S$ such that $\langle x^*,z-c\rangle<0$. Let 
    $\hat{\#}_{\rm violated}\leftarrow
    \#_{\rm violated}(A,x^*)+\Lap(\frac{1}{\epsilon})$.
    If $\hat{\#}_{\rm violated}>\Gamma+\frac{1}{\epsilon}\ln(\frac{1}{\beta})$ then halt and return ${\rm GoAllUp}(c)$. \qquad \gray{\% If there are too many violated constraints, then the private LP algorithm failed, meaning that the LP we defined was not feasible, meaning that $c\in\Conv(S)$.}

    \item\label{step:deleteLP} Delete from $S$ all points $z\in S$ such that $\langle x^*,z-c\rangle<0$. \qquad \gray{\% The private LP algorithm is allowed to violate a ``small'' number of constraints even when the LP is feasible. We delete them from the rest of the execution.}

    \item\label{step:cutElipso} Let $h$ be the halfspace $\{z\in\R^q\,:\, \langle x^*,z-c\rangle\geq0\rangle\}$. 
    Set $K\leftarrow\EEE\cap h$. Let $\hat{\EEE}\subseteq\R^q$ be the minimal (in volume) ellipsoid enclosing $K$. Let $\hat{c},\{v_i\}_{i\in[q]},\{\lambda_i\}_{i\in[q]}$ denote the new center, principal axes, and radii.

    \item Let $c\in\XXX^q$ be obtained from $\hat{c}\in\R^q$ by rounding it to the grid $\YYY^q$, and denote $\gamma=\frac{1}{4d^2}$. Set
    $\EEE\leftarrow\EEE\left(c,\{v_i\}_{i\in[d]},\{(1+\gamma)\cdot\lambda_i\}_{i\in[d]}\right)$. \qquad \gray{\% We snap the ellipsoid's center to the grid and slightly inflate its radii s.t.\ it contains the previous ellipsoid. The center of the new ellipsoid participates in the LP in the next iteration (in Step~\ref{step:callDPLP}), and so this rounding allows us to bound the precision of the constraints in this LP.}
\end{enumerate}
\item\label{step:findSubSpace} Run Algorithm~\ref{alg:affine_hulls} to privately approximate the affine span of $S$. Let $u_1,\dots,u_k\in\R^q$ denote the returned collection of affinely independent vectors (where $k\leq q$). \qquad\gray{\% If we have reached this step, then the ellipsoid method did not converge, meaning that $\Conv(S)$ is not full dimensional, and we need to restart in a lower subspace.}

\item\label{step:deleteAff} Delete from $S$ all points that do not belong to the affine span of $u_1,\dots,u_k$.

\item Run $\texttt{Aff2Lin}$ on $(u_1,\dots,u_k)$ to obtain the functions 
${\rm Project}:\R^q\rightarrow\R^{k-1}$ and ${\rm GoUp}:\R^{k-1}\rightarrow\R^q$. 

\item Set $q\leftarrow k-1$, 
replace every $x\in S$ with ${\rm Project}(x)$, 
set ${\rm GoAllUp}(\cdot)\leftarrow{\rm GoAllUp}\left({\rm GoUp}(\cdot)\right)$, and GOTO Step 2. \quad\gray{\% We compose the current function with ${\rm GoUp}$.}

\end{enumerate}
\end{algorithm}

\begin{theorem}[\cite{grunbaum2003convex}]\label{thm:polytope}
Any polytope may be defined as the convex hull of a finite set of points
(known as a {\em V-representation}), or as a bounded intersection of finitely many closed
halfspaces (known as an {\em H-representation}). In the case of a full-dimensional polytope, the
minimal such descriptions are in fact unique: The minimal V-representation
of a polytope is given by the vertices, while the minimal H-representation
consists of the facet-defining halfspaces.
\end{theorem}

\begin{lem}[Discrete separation theorem]\label{lem:Separation}
Let $Y,d\in\N$ be  parameters, let $\YYY=\Big\{\frac{a}{Y} \;:\; a\in[-Y,Y]\cap\Z\Big\}$, and let $\YYY^d$ be a $d$-dimensional grid. Let $W\subseteq\YYY^d$ be a finite set of points in $\YYY^d$. If $\vec{0}\notin\Conv(W)$ then there exists a unit vector $\bar{y}\in\R^d$ such that for every $v\in W$ we have
$$
\langle \bar{y},\bar{v} \rangle\geq 
\frac{1}{d \, (d!)^{2d+2} \, (2Y)^{2d^2 + 2d} \, Y^d}
$$
\end{lem}

\begin{proof}
We begin with the case where $\Conv(W)$ is full-dimensional.
Let $H$ be the finite collection of closed
halfspaces corresponding to the facets of $\Conv(W)$.
By Theorem~\ref{thm:polytope}, the intersection of all these halfspaces is exactly the convex hull of $W$. Hence, since $\vec{0}\notin\Conv(W)$, there must exist a halfspace $h\in H$ such that $\vec{0}\notin h$ (otherwise $\vec{0}$ belongs to all these halfspaces and hence belongs to their intersection which is the convex hull of $W$). Now let $\bar{y}$ denote the normal unit vector of $h$ and let $b$ denote its bias, so that $h=h_{\bar{y},b}=\{z\in\R^d\;:\;\langle \bar{y},z \rangle\geq b\}$. Then for every $v\in W$ we have $\langle \bar{y},v \rangle\geq b>\langle \bar{y},\vec{0} \rangle=0$. Thus, in order to show that  $\bar{y}$ satisfies the properties stated in the lemma, we next lower bound $b$.

Recall that $h$ is a halfspace that defines one of the facets of $\Conv(W)$. Thus, it can be expressed as the affine span of $d$ affinely independent points from $W$, call them $p_1,\dots,p_d\in W$. The vector $\bar{y}$ is then obtained by normalizing a solution $y$ to the set of linear equations $\langle (p_d-p_i),y \rangle=0$ for $i\in[d-1]$ (this gives $d-1$ equations with $d$ variables, and we can add the equation $\langle e_1,y \rangle=1$ as a scaling condition\footnote{technically, $\langle e_1,y \rangle=1$ might not be linearly independent with the previous equations, in which case we replace $e_1$ with some $e_j\neq e_1$ (at least one of these $d$ possible scaling equations must be independent of the other $d-1$ equations)}). Using Cramer’s rule we can express the $k$th coordinate of the vector $y$ as
$$
y[k] = \frac{{\rm det}(A_k)}{{\rm det}(A)},
$$
where $A$ is the matrix whose rows are $(p_d-p_1),\dots,(p_d-p_{d-1}),e_1$, and $A_k$ is the same as $A$ except that the $k$th column is replaced with $(0,\dots,0,1)$. Note that the entries of both $A$ and $A_k$ are numbers of the form $\frac{a}{Y}$ for integers $-2Y\leq a\leq2Y$. Thus, 
$$
{\rm det}(A),{\rm det}(A_k)\in
\Big\{\frac{a}{Y^d} \;:\; a\in[-d! (2Y)^d,d! (2Y)^d]\cap\Z\Big\}.
$$
Hence,
$$
y[k]\in
\Big\{\frac{a}{b} \;:\; a,b\in[-(d!)^2 (2Y)^{2d},(d!)^2 (2Y)^{2d}]\cap\Z\Big\}.
$$
Now let $v\in W$. We have established that $\langle \bar{y},v \rangle>0$. Thus, by our bounds on the minimal and maximal possible coordinates in $y,v$ we have
\begin{eqnarray*}
\langle \bar{y},\bar{v} \rangle &=& 
\frac{\langle y,v \rangle}{\|y\|_2\cdot\|v\|_2}\\
&\geq&\frac{\|y\|_2\cdot\|v\|_2}{\sqrt{d}(d!)^2 (2Y)^{2d}\cdot \sqrt{d}\frac{Y}{Y}}\\
&\geq&\frac{\left(\frac{1}{(d!)^2 (2Y)^{2d}}\cdot\frac{1}{Y}\right)^d}{\sqrt{d}(d!)^2 (2Y)^{2d}\cdot \sqrt{d}\frac{Y}{Y}}\\
&=&\frac{1}{d \, (d!)^{2d+2} \, (2Y)^{2d^2 + 2d} \, Y^d}\;\;,
\end{eqnarray*}
where the second inequality is because $\langle y,v \rangle$ is the sum of $d$ terms whose common denominator is at least $\left((d!)^2 (2Y)^{2d}\cdot Y\right)^d$.

\medskip
It remains to deal with the case where $\Conv(W)$ is not full-dimensional. It is still true (by Theorem~\ref{thm:polytope}) that there exists a finite collection $H$ of closed halfspaces such that their intersection is exactly $\Conv(W)$ (but now these halfspaces might not exactly correspond to the facets $\Conv(W)$). Still, as before, there must exists a halfspace $h\in H$ such that $\vec{0}\notin h$. Now we simply add points from $\YYY\cap h$ to $W$ such its convex hull becomes full-dimensional, and we are back to the previous case.
\end{proof}

\begin{lem} \label{lem:simplex}
Let $X,d\in\N$ be parameters, let $\XXX=\Big\{\frac{a}{X} \,:\, a\in[-X,X]\cap\Z\Big\}$, and let $\XXX^d$ be a $d$-dimensional grid. Let $S\subseteq\XXX^d$ be a set of $d+1$ affinely independent (i.e.\ $\Vol(\Conv(S)) >0$) points in $\XXX^d$.
Then $\Vol(\Conv(S)) \ge \frac{1}{d! X^d}$.
\end{lem}
\begin{proof}
Let $S=\{v_1,\ldots,v_{d+1} \}$.
    The volume of  $\Conv(S)$ is given by the formula
\[
    \frac{1}{d!}
    \begin{vmatrix}
    1 & v_{1,1} & \cdots & v_{1,d} \\
    1 & v_{2,1} & \cdots & v_{2,d} \\
    & & \vdots & \\
    1 & v_{d+1,1} & \cdots & v_{d+1,d}
    \end{vmatrix} \ .
    \]
    Since $S$ are affinely independent this determinant is not zero. Furthermore, since the denominator of each $v_{i,j}$ which is not zero, it is $X$, the product of the elements of each non-zero product in the expansion of the determinant is $X^d$.
\end{proof}

The following lemma shows that if our data set is full-dimensional then any ellipsoid of bounded principal radii cannot have a principal radii which is too small.

\begin{lem} \label{lem:UB2LBlambda}
Let $X,d\in\N$ be  parameters, let $\XXX=\Big\{\frac{a}{X} \;:\; a\in[-X,X]\cap\Z\Big\}$, and let $\XXX^d$ be a $d$-dimensional grid. Let $S\subseteq\XXX^d$ be a set of $d+1$ affinely independent points in $\XXX^d$.
Let $\EEE = \EEE\left(c,\{v_i\}_{i\in [d]},\{\lambda_i\}_{i\in [d]}\right)$ be an ellipsoid containing $S$,
where $\max_{i\in[d]}\lambda_i \le \Lambda$. Then 
$\min_{i\in[d]}\lambda_i \ge \frac{1}{6\,d!\,X^d \, \Lambda^{d-1}}$.
\end{lem}
\begin{proof}
The volume of an ellipsoid $\EEE = \EEE\left(c,\{v_i\}_{i\in [d},\{\lambda_i\}_{i\in [d}\right)$ equals to
\[\Vol(B_d) \lambda_1,\lambda_2\cdots\lambda_d \ , \]
where $B_d$ is a $d$-dimensional unit ball.
Since $\EEE$ contains $\Conv(S)$ we get by Lemma \ref{lem:simplex} that 
\[\Vol(B_d) \lambda_1,\lambda_2\cdots\lambda_d \ge \frac{1}{d!\,X^d } \ . \]
Rearranging we get that
\[\lambda_i\ge \frac{1}{d!\, X^d\, \Vol(B_d)
\lambda_1\cdots\lambda_{i-1}\lambda_{i+1}\cdots \lambda_d } \ . \]
The bound in the statement now follows 
by using the assumption that $\max_{i\in[d]}\{\lambda_i\} \le \Lambda$, together with the 
fact that $\Vol(B_d)\leq6$ for all $d$ (the volume of a $d$-dimensional ball vanishes as $d$ grows).

\end{proof}

We use the following theorem of \cite{john1948extremum} to prove by induction that all our Ellipsoids have bounded principal radii.

\begin{theorem}[\cite{john1948extremum}]\label{thm:john}
For every convex body $K\subseteq\R^d$ there exists a unique ellipsoid $\EEE$
of minimal volume containing $K$. Moreover, $K$ contains the ellipsoid obtained
from $\EEE$ by shrinking it from its center by a factor of $d$.
\end{theorem}

\begin{claim}\label{claim:UBlambda}
Let $\{\lambda_i\}_{i\in[q]}$ denote the radii of one of the ellipsoids constructed during the execution of Algorithm \texttt{PinHull}. Then $\max_{i\in[q]}\{\lambda_i\}\leq \sqrt{d}(2d)^T$. Here $T=\Theta(d^2\ln(dX))$ is the number of iterations from step 4 of Algorithm  \texttt{PinHull}.
\end{claim}

\begin{proof}
Let $\EEE$ denote the ellipsoid at the beginning of an iteration of Step~\ref{step:ElipLoop}. As in Step~\ref{step:cutElipso}, let $K$ be the intersection of $\EEE$ with the halfspace $h$ (whose defining hyperplane passes through $\EEE$'s center), and let $\hat{\EEE}$ be the minimal ellipsoid enclosing $K$. Let $\{\hat{\lambda}_i\}_{i\in[q]}$ denote the principal radii of $\hat{\EEE}$. Thus, by Theorem~\ref{thm:john} we have that 
$$2\max_{i\in[q]}\{\hat{\lambda}_i\}
=
\diam\left(\hat{\EEE}\right)
\leq
d\diam(K)
\leq
d\diam(\EEE)
=
2d\max_{i\in[q]}\{\lambda_i\}.$$

The next ellipsoid is obtained from $\hat{\EEE}$ by inflating each of its principal radii by a factor of $(1+\gamma)$. Thus, letting $\{\lambda_i^{\rm end}\}_{i\in[q]}$ denote the principal radii of the ellipsoid at the end of the iteration, we have that
$$
\max_{i\in[q]}\{\lambda_i^{\rm end}\} \leq (1+\gamma)d \max_{i\in[q]}\{\lambda_i\}\leq 2d\max_{i\in[q]}\{\lambda_i\}.
$$

Finally, recall that before entering the loop in Step~\ref{step:ElipLoop}, the initial ellipsoid is a ball of radius $\sqrt{q}\leq\sqrt{d}$. Thus, as there are at most $T$ iterations, we get that all the principal radii defined throughout the execution satisfy
$$
\max_{i\in[q]}\{\lambda_i\} \leq \sqrt{d}(2d)^T.
$$
\end{proof}

\begin{lem}\label{lem:centerShift}
Let $\EEE_1=\EEE\left(c,\{v_i\}_{i\in[d]},\{\lambda_i\}_{i\in[d]}\right)$ be an ellipsoid in $\R^d$. Let $0\leq\gamma\leq1$, let $\eta\in\R^d$ be s.t.\ $\|\eta\|_{\infty}\leq\frac{\gamma}{6d^{1.5}}\cdot\min_{i\in[d]}\{\lambda_i\}$, and let $\EEE_2=\EEE\left(c+\eta,\{v_i\}_{i\in[d]},\left\{\left(1+\gamma\right)\lambda_i\right\}_{i\in[d]}\right)$ be the ellipsoid obtained from $\EEE_1$ by shifting its center by $\eta$ and blowing up its principal radii by $(1+\gamma)$. Then $\EEE_1\subseteq\EEE_2$.
\end{lem}

\begin{proof}
Let $x=c+\sum_{i\in[d]}\alpha_i v_i$ be a point in $\EEE_1$ (satisfying $\sum_{i\in[d]}(\frac{\alpha_i}{\lambda_i})^2\leq1$). To show that $x\in\EEE_2$ we need to specify coefficients $\{\alpha'_i\}_{i\in[d]}$ that satisfy the following two conditions:
$$
x=c+\eta+\sum_{i\in[d]}\alpha'_i v_i 
\qquad\text{and}\qquad 
\sum_{i\in[d]}\left(\frac{\alpha'_i}{(1+\gamma)\lambda_i}\right)^2\leq1.
$$
We choose $\alpha'_i=\alpha_i-\eta_i$, where $\eta_i$ is the $i$th coefficient of $\eta$ as a linear combination of $v_1,\dots,v_d$ (the principal axes of the ellipsoid). That is, $\eta=\eta_1 v_1+\dots+ \eta_d v_d$. With this choice of $\{\alpha'_i\}$, we get that the left condition above is trivial. Indeed,
$$
c+\eta+\sum_{i\in[d]}\alpha'_i v_i = 
c+\eta+\sum_{i\in[d]}(\alpha_i-\eta_i) v_i = 
c+\sum_{i\in[d]}\alpha_i v_i = x.
$$
As for the right condition, 
observe that since $\|\eta\|_{\infty}\leq\frac{\gamma}{6d^{1.5}}\cdot\min_{i\in[d]}\{\lambda_i\}$ then for every $i\in[d]$ it holds that  
$\eta_i\leq\frac{\gamma}{6d}\cdot\min_{i\in[d]}\{\lambda_i\}$. Also recall that $\alpha_i\leq\lambda_i$ for all $i\in[d]$. We calculate:
\begin{eqnarray*}
\sum_{i\in[d]}\left(\frac{\alpha'_i}{(1+\gamma)\lambda_i}\right)^2
&=& \sum_{i\in[d]}\left(\frac{\alpha_i-\eta_i}{(1+\gamma)\lambda_i}\right)^2\\
&=& \sum_{i\in[d]}\left[\frac{\alpha^2_i}{(1+\gamma)^2\lambda^2_i}-\frac{2\cdot\alpha_i\cdot \eta_i}{(1+\gamma)^2\lambda^2_i} + \frac{\eta^2_i}{(1+\gamma)^2\lambda^2_i}\right]\\
&\leq& \sum_{i\in[d]}\left[\frac{\alpha^2_i}{(1+\gamma)^2\lambda^2_i} + \frac{2\cdot \lambda_i \cdot \frac{\gamma}{6d}\lambda_i}{(1+\gamma)^2\lambda^2_i} + \frac{(\frac{\gamma}{6d}\lambda_i)^2}{(1+\gamma)^2\lambda^2_i}\right]\\
&\leq& \sum_{i\in[d]}\left[\frac{\alpha^2_i}{(1+\gamma)^2\lambda^2_i} + \frac{\gamma}{3d} + \frac{\gamma}{6d}\right]
\;=\;
\frac{1}{(1+\gamma)^2}\sum_{i\in[d]}\left[\frac{\alpha^2_i}{\lambda^2_i} \right]+\frac{\gamma}{2}\\
&\leq&
\frac{1}{(1+\gamma)^2}+\frac{\gamma}{2}
\leq
\left(1-\frac{\gamma}{2}\right)+\frac{\gamma}{2} \leq 1.
\end{eqnarray*}
\end{proof}

Combining Claim~\ref{claim:UBlambda} with Lemma~\ref{lem:UB2LBlambda} and 
Lemma~\ref{lem:centerShift} yields the following claim.

\begin{claim}\label{claim:enclosing}
Let $S^{\rm begin},S^{\rm end}$ and $\EEE^{\rm begin},\EEE^{\rm end}$ denote the multiset $S$ and the ellipsoid $\EEE$ as they are at the beginning and at the end of one of the iterations of Step~\ref{step:ElipLoop}. If 
$\Conv\left(S^{\rm begin}\right)\subseteq \EEE^{\rm begin}$ and 
$\Vol\left(\Conv\left(S^{\rm end}\right)\right)>0$, then $\Conv\left(S^{\rm end}\right)\subseteq\EEE^{\rm end}$.
\end{claim}

\begin{proof}
Let $h,K,\hat{\EEE}$ denote 
halfspace, the convex body, and the ellipsoid from Step~\ref{step:cutElipso} (recall that $\hat{\EEE}$ is transformed into $\EEE^{\rm end}$ by shifting and inflating). By construction, and by the assumption that $\Conv\left(S^{\rm begin}\right)\subseteq \EEE^{\rm begin}$, we have that
$$S^{\rm end}=S^{\rm begin}\cap h\subseteq \EEE^{\rm begin}\cap h = K \subseteq\hat{\EEE}.$$
As $\hat{\EEE}$ is convex we get that $\Conv\left(S^{\rm end}\right)\subseteq\hat{\EEE}$.
Thus, combining Claim~\ref{claim:UBlambda} and Lemma~\ref{lem:UB2LBlambda} with the assumption that 
$\Vol\left(\Conv\left(S^{\rm end}\right)\right)>0$, we get that the the radii of $\hat{\EEE}$ satisfy
$$\min_{i\in[d]}\hat{\lambda}_i \ge \frac{1}{6\,d!\,X^d \, d^{d/2} \, (2T)^{dT}}\triangleq\hat{\lambda}_{\rm min}$$

Therefore, 
by asserting that 
$$Y\geq \frac{6d^{1.5}}{\gamma\hat{\lambda}_{\rm min}}=
144 d^{3.5}  d!\,X^d \, d^{d/2} \, (2T)^{dT},
$$
we make sure that the centers of $\EEE^{\rm end}$ and $\hat{\EEE}$ satisfy 
$$\|c-\hat{c}\|_{\infty}\leq\frac{1}{Y}\leq
\frac{\gamma}{6d^{1.5}}\cdot\min_{i\in[d]}\{\lambda_i\}.
$$ 
Thus, by Lemma~\ref{lem:centerShift}, we have that $\hat{\EEE}\subseteq\EEE^{\rm end}$, so overall $\Conv\left(S^{\rm end}\right)\subseteq\hat{\EEE}\subseteq\EEE^{\rm end}$.
\end{proof}

By induction on the number of iterations of Step~\ref{step:ElipLoop}, Claim~\ref{claim:enclosing} results in the following corollary.

\begin{cor}\label{cor:enclosing-simplified}
Let $S^{\rm end}$ and $\EEE^{\rm end}$ denote the multiset $S$ and the ellipsoid $\EEE$ as they are at the end of one of the iterations of Step~\ref{step:ElipLoop}. If $\Vol\left(\Conv\left(S^{\rm end}\right)\right)>0$ then $\Conv\left(S^{\rm end}\right)\subseteq\EEE^{\rm end}$.
\end{cor}

\begin{claim}\label{claim:volShrink}
Every iteration of Step~\ref{step:ElipLoop} decreases the volume of the current ellipsoid multiplicatively by at least $e^{-\frac{1}{4d}}$.
\end{claim}

\begin{proof}
Let $\EEE^{\rm begin}$ denote the ellipsoid at the beginning the iteration. As in Step~\ref{step:cutElipso}, let $K$ be the intersection of $\EEE^{\rm begin}$ with the halfspace $h$ (whose defining hyperplane passes through $\EEE^{\rm begin}$'s center), and let $\hat{\EEE}$ be the minimal ellipsoid enclosing $K$. Recall that, as in the standard analysis of the (non-private) ellipsoid method,
$$\Vol\left(\hat{\EEE}\right)\leq\Vol\left(\EEE^{\rm begin}\right)\cdot \exp\left(-\frac{1}{2q}\right)\leq\Vol\left(\EEE^{\rm begin}\right)\cdot \exp\left(-\frac{1}{2d}\right).$$
The ellipsoid $\EEE^{\rm end}$ at the end of the iteration is obtained from $\hat{\EEE}$ by inflating each of its radii by a $(1+\gamma)$ factor (and by shifting it, but this has no effect on its volume). This inflates the volume by $(1+\gamma)^q\leq(1+\gamma)^d\leq e^{\gamma d}$. Overall, 
$$\Vol(\EEE^{\rm end})\leq\Vol\left(\EEE^{\rm begin}\right)\cdot \exp\left(-\frac{1}{2d}\right)\cdot\exp(\gamma d)
\leq
\Vol\left(\EEE^{\rm begin}\right)\cdot \exp\left(-\frac{1}{4d}\right),
$$
where the last inequality follows by asserting that $\gamma\leq\frac{1}{4d^2}$
\end{proof}

\begin{definition}
Denote 
$\Gamma=\Theta
\left(
\frac{d^{12.5}}{\epsilon}
 \ln^{1.5}(d) \ln^2\left(\frac{d\ln(X)}{\beta}\right)  
\ln^{1.5}(dX)\ln^{1.5}(d\ln(X)) 
 \ln\left(\frac{d\ln(X)
 }{\beta\delta}\right)
 \sqrt{\ln\frac{1}{\delta}}
 \right)$.\footnote{The value for $\Gamma$ here is taken from Theorem~\ref{thm:vempala} after plugging in our choice for the roundness parameter $\rho$, as defined in Step~\ref{step:callDPLP} of \texttt{PinHull}.} We say that a call to the private LP algorithm in Step~\ref{step:callDPLP} {\em fails} on an input matrix $A$ if $\rho(A)\geq\rho$ and the returned vector $x^*$ satisfies $\#_{\rm violated}(A,x^*)>\Gamma$. That is, the call {\em fails} if the returned solution violates many constraints, even though the given LP is feasible with large roundness. Otherwise, we say that the call {\em succeeded}.

 Similarly, we say that a call to Algorithm~\ref{alg:affine_hulls} on input $S$ made from Step~\ref{step:findSubSpace} of Algorithm \texttt{PinHull} {\em succeeds} if the returned (affinely independent) vectors $u_1,\dots,u_k$ satisfy $\Affspan(u_1,\dots,u_k)\subseteq\Affspan(S)$ and $|\{x\in S : x\notin\Affspan(u_1,\dots,u_k)\}|\leq O(\frac{d^2}{\epsilon}\log(\frac{d}{\beta\delta}))$.
\end{definition}

\begin{definition}
Let $E$ denote the event that, throughout the execution of Algorithm \texttt{PinHull}, all calls to the private LP algorithm and to Algorithm~\ref{alg:affine_hulls} succeed, and all Laplace samples from Step~\ref{step:verifyDPLP} are at most $\frac{1}{\epsilon}\log(\frac{1}{\beta})$ in absolute value.
\end{definition}

The following observation follows from a union bound over the failure probabilities of all the executions of the private LP algorithm and Algorithm~\ref{alg:affine_hulls} (as captured by Theorems~\ref{thm:vempala} and~\ref{t:privateaffspan}), and by a standard tail bound for the Laplace distribution. Note that each of these mechanisms is applied at most $dT$ times throughout the execution.

\begin{obs}
$\Pr[E]\geq1-3\beta dT$.
\end{obs}

We proceed with the analysis assuming that Event $E$ indeed holds.

\begin{claim}\label{claim:T}
If the algorithm did not halt during the $T$ iterations of the loop in Step~\ref{step:ElipLoop}, then at the end of this loop the remaining input points, call them $S^{\rm end}$, are affinely dependent, i.e., $\Conv\left(S^{\rm end}\right)$ has zero volume in $\R^q$.
\end{claim}

\begin{proof}
Assume the converse towards contradiction. By Lemma~\ref{lem:simplex}, we have that $\Vol\left(\Conv\left(S^{\rm end}\right)\right)\geq\frac{1}{d!\, X^d}$. Let $\EEE^{\rm end}$ denote the ellipsoid $\EEE$ at the end of the loop. Then, by Corollary~\ref{cor:enclosing-simplified}, we have that $\Conv\left(S^{\rm end}\right)\subseteq\EEE^{\rm end}$ and thus $\Vol\left(\EEE^{\rm end}\right)\geq\frac{1}{d!\, X^d}$.

On the other hand, the volume of the initial ellipsoid (before starting the loop), which is a ball, is at most $6 d^{d/2}$, and by Claim~\ref{claim:volShrink} the volume shrinks by a factor of at least $e^{-\frac{1}{4d}}$ with every iterations. Thus, after $T$ iterations we have that $\Vol\left(\EEE^{\rm end}\right)\leq 6 d^{d/2}\cdot e^{-\frac{T}{4d}}$. For $T\geq \Omega(d^2\ln(dX))$ this is less than $\frac{1}{d!\, X^d}$, giving the desired contradiction.
\end{proof}

\begin{lem}\label{lem:sampleComplexityPinHull}
If the size of the initial dataset satisfies
$|S|\geq \Omega(dT\Gamma)$,
then Algorithm \texttt{PinHull} returns a point in $\Conv(S)$ whenever Event $E$ occurs.
\end{lem}

\begin{proof}
Note that there are two places in which the algorithm might halt its execution: Either in Step~\ref{step:verifyDPLP} or in Step~\ref{step:base} (the base case of the algorithm).

We first show that whenever the algorithm halts in Step~\ref{step:verifyDPLP} then it returns a valid solution (a point in the convex hull of its input points). To this end, recall that in Step~\ref{step:callDPLP} we run the private LP algorithm to try and separate $c$ (the center of the current ellipsoid) from the (remaining) input points. By the discrete separation theorem (see Lemma~\ref{lem:Separation}), if $c$ is not in the convex hull of remaining input points then this LP is feasible with roundness $\rho$. In this case, assuming that Event $E$ occurs, then the private LP algorithm succeeds in finding a solution $x^*$ that violates at most $\Gamma$ constraints. Therefore, again assuming that $E$ occurs, in Step~\ref{step:verifyDPLP} we get that the noisy estimation for the number of violated constraints, $\hat{\#}_{\rm violated}$, would be small and the algorithm would {\em not} halt. This shows that if the algorithm does halt in Step~\ref{step:verifyDPLP} then it must be that the center of the current ellipsoid, $c$, belongs to the convex hull of the remaining input points in $\R^q$. Thus, by Lemma~\ref{lem:aff2lin}, the returned point ${\rm GoAllUp}(c)$ belong to the convex hull of the original input points in $\R^d$.

We next show that
if the input dataset is large enough, then
the algorithm returns a valid solution also when it halts in Step~\ref{step:base} (the base case in $\R^0$). To show this, it suffices to show that the algorithm never reaches Step~\ref{step:base} with an empty dataset. Indeed, if the algorithm reaches Step~\ref{step:base} with a non-empty dataset, then all points in this dataset are equal to the empty tuple, and thus the empty tuple $c=()$ is in the convex hull of the remaining points in $\R^0$. Again, in this case, by Lemma~\ref{lem:aff2lin}, the returned point ${\rm GoAllUp}(c)$ belong to the convex hull of the original input points in $\R^d$.

It thus remains to show that the dataset never becomes empty throughout the execution. Let us examine the steps during which input points get deleted. This can happen in the following two steps of the algorithm:
\begin{itemize}
    \item In Step~\ref{step:deleteLP}, which is executed at most $dT$ times, we delete at most $\Gamma+\frac{1}{\epsilon}\ln(\frac{dT}{\beta})$ from $S$.
    \item 
 In Step~\ref{step:deleteAff}, which is executed at most $d$ times, we delete at most $O(\frac{d^2}{\epsilon}\log(\frac{d}{\beta\delta}))$ points from $S$.
\end{itemize}
Overall, at most $O(dT\Gamma)$ pints are deleted from $S$ throughout the execution. Thus, assuming that the size of the initial dataset satisfies $|S|\geq\Omega(dT\Gamma)$, then the dataset never becomes empty throughout the execution.
\end{proof}

\begin{lem}\label{lem:privacyPinHull}
Denote 
$k=3dT=\Theta(d^3\ln(dX))$. 
Algorithm \texttt{PinHull} satisfies $\left(\hat{\epsilon},2k\delta\right)$-differential privacy for 
$\hat{\epsilon}=\sqrt{2k\ln(\frac{1}{k\delta})}\cdot\epsilon + 2k\epsilon^2$.
\end{lem}

\begin{proof}
Follows from standard composition theorems for differential privacy \citep{DRV10}, together with the fact that throughout the execution the algorithm accesses its input database via at most $dT$ applications of the private LP algorithm (in Step~\ref{step:callDPLP}), at most $dT$ applications of the Laplace mechanism (in Step~\ref{step:verifyDPLP}), and at most $d$ applications of Algorithm~\ref{alg:affine_hulls} (in Step~\ref{step:findSubSpace}). Each of these mechanisms satisfies $(\epsilon,\delta)$-differential privacy.
\end{proof}

By combining Lemma~\ref{lem:sampleComplexityPinHull} and Lemma~\ref{lem:privacyPinHull}, we establish the following theorem.
\begin{customthm}{E}[Point in Convex Hull]\label{thm:pich}
    Denote 
$k=3dT=\Theta(d^3\ln(dX))$. 
Algorithm \texttt{PinHull} satisfies $\left(\hat{\epsilon},2k\delta\right)$-differential privacy for 
$\hat{\epsilon}=\sqrt{2k\ln(\frac{1}{k\delta})}\cdot\epsilon + 2k\epsilon^2$. Furthermore, when given an input dataset $S\in\XXX^d$ of size at least
$$
|S|\geq\Theta(dT\Gamma)=
\Theta
\left(
\frac{d^{15.5}}{\epsilon}
 \ln^{1.5}(d) \ln^2\left(\frac{d\ln(X)}{\beta}\right)  
\ln^{2.5}(dX)\ln^{1.5}(d\ln(X)) 
 \ln\left(\frac{d\ln(X)
 }{\beta\delta}\right)
 \sqrt{\ln\frac{1}{\delta}}
 \right)
$$
then with probability at least $1-3\beta dT$ the algorithm returns a point in the convex hull of $S$.
\end{customthm}

By appropriately rescaling the parameters $\beta,\epsilon,\delta$ in this theorem we get the following result.

\begin{cor}
Let $X,d\in\N$ be parameters, let $\XXX=\Big\{\frac{a}{X} \;:\; a\in[-X,X]\cap\Z\Big\}$, and let $\XXX^d$ be a $d$-dimensional grid. There exists a computationally efficient $(\epsilon,\delta)$-differentially private algorithm that takes a dataset $S\subseteq\XXX^d$ and return a point in the convex hull of $S$ with probability at least $1-\beta$, provided that
\begin{eqnarray*}
|S|&\geq&\Theta
\left(
\frac{d^{17}}{\epsilon}
 \ln^{1.5}(d) \ln^2\left(\frac{d\ln(X)}{\beta}\right)  
\ln^{3}(dX)\ln^{1.5}(d\ln(X)) 
 \ln^2\left(\frac{d\ln(X)
 }{\beta\delta}\right)
 \right)\\
&=&\frac{d^{17}}{\epsilon}\cdot{\rm polylog}\left(d,\frac{1}{\beta},\frac{1}{\delta},X\right).
\end{eqnarray*}

\end{cor}

\noindent
\section*{Acknowledgment}
We thank Eden Chlamt{\'a}{\v{c}} for communicating to us the idea of the reduction in Section \ref{sec:no-margin} as it applies in a non-private setting.

\appendix

\section{Learning Subspaces}\label{sec:pacsubspace}

As an application we derive a PAC learning algorithm for subspaces in $\mathbb{R}^d$.
We use standard learning theoretic notation (see, e.g., \cite*{SSbook}).
    Let $X$ be a set called the domain and let $Y=\{0,1\}$ denote the label-set.
    For a set $Z$, let $Z^\star :=\cup_{n=0}^\infty Z^n$ denote the space of all finite sequences with elements from $Z$. An {\em hypothesis (or classifier)} is a function $h:X\to Y$. 
    An {\em example} is a pair $z=(x,y)\in X\times Y$.
    A {\em sample} $S\in (X\times Y)^\star$ is a (finite) sequence of examples.
    A {\em learning rule} is a mapping from $(X\times Y)^\star$ to $Y^X$;
    i.e., the input is a finite sample and the output is a hypothesis.
    Given a distribution $D$ over $X\times Y$ and an hypothesis $h$, 
    the {\em (population) loss} of $h$ with respect to $D$ 
    is $L_D(h) = \E_{(x,y)\sim D}1[h(x)\neq y]$.
    Given a sample $S=\{(x_i,y_i)\}_{i=1}^m$, 
    the {\rm (empirical) loss} of~$h$ with respect to $S$
    is  $L_S(h) =\frac{1}{m}\sum_{i=1}^m{1[h(x_i)\neq y_i]}$, where $1[\cdot]$ is the indicator function.

We use the following basic property of differentially private learning rules.
\begin{theorem}[Privacy $\implies$ Generalization (see e.g.\ \cite{AlonBMS20}, Theorem 3.14)]\label{t:privgen}
    Let $A$ be an $(\eps,\delta)$-differentially private learning rule, let $D$ be a distribution
    over examples, and let $m\in\mathbb{N}$. Then,
    \[
    \Pr_{S\sim D^m}\Bigl[L_D(h) \geq e^{2\eps}\Bigl(L_S(h) + \frac{\frac{20}{\eps}\log\frac{1}{\delta} }{ m}\Bigr)\Bigr]\leq O(\delta),
    \]
    where $h=A(S)$ and the \(O(\cdot)\) notation conceals a universal numerical constant.
\end{theorem}

\begin{algorithm}
\caption{Private Subspace Learner}
\label{alg:subspacelearn}
\textbf{Input:} A sequence \(S=\{(\boldsymbol{x_i},y_i)\}_{i=1}^m\) of labeled examples, where \(\boldsymbol{x_i} \in \mathbb{F}^d\), and \(y_i = 1[\boldsymbol{x_i} \in V^\star]\), with \(V^\star \subseteq \mathbb{F}^d\) being an unknown affine subspace.\\
\textbf{Output:} An hypothesis function \(h: \mathbb{F}^d \to \{0,1\}\) such that \(h(\boldsymbol{x})\) closely approximates \(1[\boldsymbol{x} \in V^\star]\). (See Theorem~\ref{t:privsubspacelearn}.)
\begin{algorithmic}[1]
\State Apply \Cref{alg:affine_hulls} to the subsequence \(\{\boldsymbol{x}_i : y_i = 1\}\) consisting of positively labeled examples.
\State \textbf{Output} the hypothesis \(h(\boldsymbol{x}) = 1[\boldsymbol{x} \in U]\), where \(U\) is the affine span\footnotemark\ of the sequence produced by \Cref{alg:affine_hulls}.
\end{algorithmic}
\end{algorithm}
\footnotetext{We use the convention that the affine span of the empty sequence is the empty set.}

\begin{customthm}{C}[Privately Learning Subspaces]\label{t:privsubspacelearn}
\Cref{alg:subspacelearn} is $(\eps,\delta)$-DP. Furthermore, if $\eps=O(1)$, say $\eps\leq 10$, then it satisfies the following.
Let $V^\star\subseteq \mathbb{F}^d$ be an affine subspace, let $D_X$ be a distribution over \(\mathbb{F}^d\) and let $D$ be a distribution over examples $(x,y)$, where $x\sim D_X$ and $y=1[\boldsymbol{x}\in V^\star]$.
Let \(S=\{(\boldsymbol{x_i},y_i)\}_{i=1}^m\) be an IID input sample drawn from $D$. 
Then, with probability at least $1-O(\delta)$ %
\[L_D(h) \leq  O\Biggl(\frac{\frac{d^2}{\eps}\log\frac{d}{\delta}}{m}\Biggr),\]
where $h$ is the hypothesis outputted by \Cref{alg:subspacelearn} on the input sequence $S\sim D^m$.
\end{customthm}
\begin{remark}
It is worth reemphasizing \Cref{rem:effectivedim}  %
in the context of \Cref{t:privsubspacelearn}.
It implies that the dimension $d$ can be substituted with $\dim(V^\star)$, which is the dimension of the space being learned. Consequently, the error rate and sample complexity scales with the dimension of the target subspace rather the dimension of the ambient space.
\end{remark}
\begin{proof}
\Cref{alg:subspacelearn} inherits the \((\eps,\delta)\)-privacy guarantees of \Cref{alg:linspan},
as follows by the post-processing property of differentially private (DP) algorithms.

For the utility, by \Cref{t:privateaffspan} we have that with probability at least $1-\delta$:
\[L_S(h) \leq O\Biggl(\frac{\frac{d^2}{\eps}\log\frac{d}{\delta}}{m}\Biggr).\]
By \Cref{t:privgen} we have that with probability at least $1-O(\delta)$:
\[L_D(h) \leq e^{2\eps}\Bigl(L_S(h) + \frac{\frac{20}{\eps}\log\frac{1}{\delta} }{ m}\Bigr). \]
Thus, by a union bound and since $\eps\leq 10$ we get that with probability at least $1-O(\delta)-\delta=1-O(\delta)$:
\[L_D(h) \leq O\Biggl(e^{2\eps}\Biggl(\frac{\frac{d^2}{\eps}\log\frac{d}{\delta}}{m} + \frac{\frac{20}{\eps}\log\frac{1}{\delta} }{ m}\Biggr)\Biggr) = O\Biggl(\frac{\frac{d^2}{\eps}\log\frac{d}{\delta}}{m}\Biggr),\]
as required.
\end{proof}

\section{A Private Version of the LP Algorithm of Dunagan and Vempala}
\label{appendix:vempala}

Consider an $n\times d$ matrix $A$ s.t.\ there exists a non-zero vector $x$ for which $Ax\geq0$. Our goal is to privately find vector $x\neq0$ for which $\langle a_i,x \rangle\geq0$ for {\em most} of the rows $a_i$ of $A$.
\cite{DS08} presented an efficient perceptron-like algorithm for solving such linear programs. A known folklore is that their algorithm can be transformed to operate in the {\em statistical queries (SQ)} model, and can therefore be transformed to preserve differential privacy via a generic result of \cite{BDMN05}. However, we are not aware of any formal statement of the resulting SQ analogue or DP analogue of the algorithm of \cite{DS08}. Hence, we include in this section a complete description of the resulting differentially private algorithm. Aside from the technical issue of keeping track of the noises added for privacy, the analysis in this section is almost identical to that of \cite{DS08}.

We use the following variant of the Gaussian mechanism of \cite{DKMMN06} to obtain noisy averages of unit vectors in $\R^d$. The straightforward way to obtain a privacy preserving estimation of an average is to add independent noise to the numerator and the denominator. The following variant of the Gaussian mechanism allows us to work with a single Gaussian noise which is added directly to the average. This will help simplify some of our arguments later on.

\begin{algorithm}[H]
\caption{\texttt{NoisyAVG}}
{\bf Input:} Multiset $V$ of unit vectors in $\R^d$, predicate $g:\R^d\rightarrow\{0,1\}$, parameters $\epsilon,\delta$.
\begin{enumerate}[rightmargin=10pt,itemsep=1pt]
\item Set $\hat{m} = |\{v\in V : g(v)=1\}| + \Lap(2/\epsilon) - \frac{2}{\epsilon}\ln(2/\delta)$. If $\hat{m}\leq 0$ then output $\bot$ and halt.
\item Denote $\sigma = \frac{4}{\epsilon\hat{m}}\sqrt{2\ln(8/\delta)}$, and let $\eta\in\R^d$ be a random noise vector with each coordinate sampled independently from $\NNN(0,\sigma^2)$. 
Return $\frac{\sum_{v\in V : g(v)=1}v}{|\{v\in V : g(v)=1\}|}+\eta$.
\end{enumerate}
\end{algorithm}

\begin{obs}[\cite{NSV16}]
Let $V$ and $g$ be s.t.\ $m=|\{v\in V : g(v)=1\}|\geq\frac{16}{\epsilon}\ln(\frac{2}{\beta\delta})$. 
With probability at least $(1-\beta)$ algorithm $\texttt{NoisyAVG}(V)$ returns $\frac{\sum_{v\in V : g(v)=1}v}{|\{v\in V : g(v)=1\}|}+\eta$ where $\eta$ is a vector whose every coordinate is sampled i.i.d.\ from $\NNN(0,\sigma^2)$ for some $\sigma\leq\frac{8}{\epsilon m}\sqrt{2\ln(8/\delta)}$
\end{obs}

\begin{theorem}[\cite{NSV16}]\label{thm:NoisyAVG}
Algorithm \texttt{NoisyAVG} is $(\epsilon,\delta)$-differentially private.
\end{theorem}

The number of iterations in the classic perceptron algorithm depends on the ``wiggle room'' available for a feasible solution, captured by the following definition:
\begin{definition}[\cite{DS08}]
Let $A$ be an $n\times d$ matrix whose rows are $a_1,\ldots,a_n$.
The {\em roundness} of the feasible region of $A$ is
$$
\rho(A)=\max_{x:\|x\|_2=1, Ax\geq0} \min_{a_i}  \langle \bar{a_i},x \rangle,
$$
where $\bar{a_i}$ is the unit vector along $a_i$.
\end{definition}
So if $\rho(A)\gg0$ then not only  is there  a vector $x$ s.t.\ $Ax\geq0$, but also a vector $z$ s.t.\ $\langle \bar{a_i},z \rangle\geq\rho(A)\gg0$ for every row $a_i$ of $A$.

\begin{algorithm}[H]
\caption{\texttt{PrivateLP} \label{alg:privateLP}}
{\bf Input:} Parameters $\beta,\epsilon,\delta$, and an $n\times d$ matrix $A$ satisfying $\rho(A)\geq\rho_0$.
\begin{enumerate}[rightmargin=10pt,leftmargin=37pt,itemsep=1pt]
\item Let $B=I,$ and let $\Delta=1/(500d)$.

\item Repeat at most $\Theta(\ln(\frac{1}{\beta}))$ times:\\

\vspace{-35pt}
\[\hspace{-17pt}\rotatebox{90}{\hspace{-45px}\gray{Perceptron improvement}}
\color{gray}\left[\color{black}\parbox{\textwidth}{
		\begin{enumerate}
		\item Let $y$ be a random unit vector in $\R^d$.
		\item Repeat at most $\frac{8}{\Delta^2}\ln(3\sqrt{d})$ times:
				\begin{enumerate}
				\item\label{step:break} Set $\hat{m}=|\{a\in A : \langle \bar{a},\bar{y} \rangle<-\Delta\}|+\Lap(1/\epsilon)$.\\
							If $y\neq0$ and $\hat{m}\leq \nu=
       \Theta\left(\frac{1}{\epsilon}d^{2.5}\ln(d)\ln(\frac{1}{\beta\delta})\right)$, then break and goto step~3.
				\item Define the predicate $g_y$ where $g_y(v)=1$ iff $\langle v,\bar{y} \rangle<-\Delta$.\\ 
							Let $u=\texttt{NoisyAVG}(\bar{A},g_y,\epsilon,\delta)$, where $\bar{A}$ is the same as $A$ after normalizing its rows.
				\item Set $y=y - \langle u,y \rangle u$.
				\end{enumerate}
		\end{enumerate}
}\right.
\]

\vspace{-10pt}
\item Delete from $A$ every vector $a$ s.t.\ $\langle \bar{a},\bar{y} \rangle<-\Delta$.

\newcounter{enumTemp}
\setcounter{enumTemp}{\theenumi}
\end{enumerate}

\vspace{-20pt}
\[\rotatebox{90}{\hspace{-30px}\gray{Perceptron phase}}
\color{gray}\left[\color{black}\parbox{\textwidth}{
\begin{enumerate}[rightmargin=10pt,topsep=0pt,parsep=0pt,leftmargin=20pt,itemsep=5pt]
\setcounter{enumi}{\theenumTemp}
\item Let $x=e_1=(1,0,0,\dots,0)$.
\item Repeat at most $\Theta(d^2)$ times:
		\begin{enumerate}
		\item Set $\hat{m}=|\{a\in A : \langle \bar{a},\bar{x} \rangle\leq\frac{\Delta}{24}\}|+\Lap(1/\epsilon)$. If $\hat{m}\leq \zeta=\Theta\left(\frac{d^2}{\epsilon}\ln(\frac{1}{\beta\delta})\right)$ then halt and return $Bx$. %
		\item Define the predicate $g_x$ where $g_x(v)=1$ iff $\langle v,\bar{x} \rangle\leq\frac{\Delta}{24}$, and let $u=\texttt{NoisyAVG}(\bar{A},g_x,\epsilon,\delta)$.
		\item Set $x=x+u$.
		\end{enumerate}
\setcounter{enumTemp}{\theenumi}
\end{enumerate}
}\right.
\]

\vspace{-20pt}
\begin{enumerate}[rightmargin=10pt,leftmargin=37pt,itemsep=1pt]
\setcounter{enumi}{\theenumTemp}
\item Set $A=A\left(I+\bar{y}\bar{y}^T\right)$ and $B=B\left(I+\bar{y}\bar{y}^T\right)$.\\
\gray{\% This is equivalent to replacing $A$'s and $B$'s rows by $a_i\leftarrow a_i+\langle a_i,\bar{y} \rangle\bar{y}$ and $b_i\leftarrow b_i+\langle b_i,\bar{y} \rangle\bar{y}$.}
\item If this step was reached more than $T=\Theta(d\cdot\ln(1/\rho_0)+\ln(1/\beta))$ times then FAIL. Otherwise GOTO Step~2.
\setcounter{enumTemp}{\theenumi}
\end{enumerate}
\end{algorithm}

\begin{definition}
Let $A$ be an $n\times d$ matrix with rows $a_1,\ldots,a_n$ and consider applying the perceptron phase of Algorithm \ref{alg:privateLP} to $A$. We say that the perceptron phase {\em succeeded} if the algorithm halts in Step 5(a) and returns a vector $Bx^*$ such that $|\{i: \langle \bar{a}_i,\bar{x}^* \rangle<\frac{\Delta}{24}\}|\leq2\zeta=\Theta\left(\frac{d^2}{\epsilon}\ln(\frac{1}{\beta\delta})\right)$. 
\end{definition}

\begin{lem}[\cite{BDMN05}]\label{lem:perceptron}
Let $n> 2\zeta =\Theta\left(\frac{d^2}{\epsilon}\ln(\frac{1}{\beta\delta})\right)$.
Let $A$ be an $n\times d$ matrix with rows $a_1,\ldots,a_n$ and consider applying the perceptron phase of Algorithm \ref{alg:privateLP} on $A$. Then,
\begin{enumerate}
\item W.p.\ at least $1-\beta$, whenever the algorithm halts in Step 5(a) then the perceptron phase succeeded.

\item If $\rho(A)\geq\Delta=1/(500d)$, then with probability at least $1-\beta$ the perceptron phase succeeds.
\end{enumerate}
\end{lem}

\begin{proof}
Part 1 of the lemma follows from a simple union bound on the magnitude of the Laplace noises sampled in Step 5(a) throughout the phase. Specifically, with probability at least $1-\beta$, all of these samples are at most $O(\frac{1}{\epsilon}\log(\frac{d}{\beta}))\ll\zeta$ in absolute value, in which case Part 1 of the lemma holds. We continue to prove Part 2 of the lemma under the assumption that this bound on the Laplace random variables indeed holds.

Let $x^i$ denote the vector $x$ as it is after the $i^\text{th}$ update of Step~5(c), and let $u^i$ be the corresponding vector s.t.\ $x^{i+1}=x^i+u^i$.
Also let $m^i=|\{a\in A : \langle \bar{a},\bar{x}^i \rangle\leq\frac{\Delta}{24}\}|$, and ${\rm avg}^i=\left(\frac{1}{m^i}\sum_{a\in A: \langle \bar{a},\bar{x}^i \rangle\leq\frac{\Delta}{24}}\bar{a}\right)$.

By our assumption on the noise magnitude in Step~5(a), in every iteration that does not halt in Step~5(a) we have that $m^i\geq\zeta/2$. Therefore, by the properties of algorithm \texttt{NoisyAVG} (and a union bound over the $O(d^2)$ iterations), with probability at least $1-\beta$, in all of those iterations we also have that $u^i={\rm avg}^i + \eta^i$ where $\eta^i$ is a vector whose every coordinate is sampled i.i.d.\ from $\NNN(0,\sigma^2)$ for some $\sigma\leq\frac{8}{\epsilon m^i}\sqrt{2\ln(8/\delta)}\leq\sigma^*=\Theta\left(\frac{1}{d^2\sqrt{\ln(\frac{1}{\beta\delta})}}\right)$. We proceed with the analysis assuming that this is the case.

Let $z$ be a unit vector s.t.\ $\min_{a_i}  \langle \bar{a_i},z \rangle=\rho(A)\geq1/(500d)$, and consider the potential function $\frac{\langle x,z \rangle}{\|x\|_2}$. The numerator increases linearly (w.h.p.) with the number of steps:

\begin{eqnarray*}
\langle x^{i+1},z \rangle &=& \langle x^i+u^i,z \rangle \\
&=& \langle x^i,z \rangle + \langle u^i,z \rangle\\
&=& \langle x^i,z \rangle + \langle {\rm avg}^i, z \rangle + \langle \eta^i,z \rangle\\
&\geq& \langle x^i,z \rangle + \frac{1}{500d} + \langle \eta^i,z \rangle,
\end{eqnarray*}
where the last inequality is because ${\rm avg}^i$ is the average of vectors $\bar{a_i}$ satisfying $\langle \bar{a_i},z \rangle\geq\rho(A)\geq\frac{1}{500d}$.
Since $z$ is a unit vector and since every coordinate of $\eta^i$ is distributed as $\NNN(0,\sigma^2)$, we have that $\langle \eta^i,z \rangle\sim\NNN(0,\sigma^2)$. Using our bound on $\sigma$, we get that $\Pr[\langle \eta^i,z \rangle <-\frac{1}{1000d}]\leq O\left(\frac{\beta}{d^2}\right)$. Thus, with probability at least $1-\beta$, every step increases $\langle x,z \rangle$ by at least $\frac{1}{1000d}$. We proceed with the analysis assuming that this is the case.

We now show that the denominator of our potential function increases much more slowly. Let us analyze its square:
\begin{eqnarray*}
 (\|x\|_2)^2 &=& \langle x^{i+1},x^{i+1} \rangle\\
&=& \langle x^i+{\rm avg}^i+\eta^i,x^i+{\rm avg}^i+\eta^i \rangle \\
&=& \langle x^i,x^i \rangle + \langle {\rm avg}^i,{\rm avg}^i \rangle +  \langle \eta^i,\eta^i \rangle  +  2\langle x^i,\eta^i \rangle + 2\langle {\rm avg}^i,\eta^i \rangle +  2\langle x^i,{\rm avg}^i \rangle \\
&=& \langle x^i,x^i \rangle + \langle {\rm avg}^i,{\rm avg}^i \rangle +  \langle \eta^i,\eta^i \rangle  +  2\langle x^i,\eta^i \rangle + 2\langle {\rm avg}^i,\eta^i \rangle +  2\|x^i\|_2 \langle \bar{x}^i,{\rm avg}^i \rangle\\
&\leq& \langle x^i,x^i \rangle + \langle {\rm avg}^i,{\rm avg}^i \rangle +  \langle \eta^i,\eta^i \rangle  +  2\langle x^i,\eta^i \rangle + 2\langle {\rm avg}^i,\eta^i \rangle + \frac{\|x^i\|_2}{6000d}\\
&\leq& \langle x^i,x^i \rangle + 1 +  \langle \eta^i,\eta^i \rangle  +  2\langle x^i,\eta^i \rangle + 2\langle {\rm avg}^i,\eta^i \rangle+ \frac{\|x^i\|_2}{6000d},
\end{eqnarray*}
where the first inequality is because ${\rm avg}^i$ is the average of vectors $\bar{a}_i$ satisfying $\langle x^i,\bar{a}_i \rangle\leq\frac{\Delta}{24}=\frac{1}{12000d}$, and the last inequality is because $\langle {\rm avg}^i,{\rm avg}^i \rangle\leq1$ by the triangle inequality (${\rm avg}^i$ is the average of unit vectors). Let us examine each of the terms in the above expression:

\begin{itemize}
    \item $\langle {\rm avg}^i,\eta^i \rangle$ is distributed as $\NNN(0,\sigma^2\|{\rm avg}^i\|^2_2)$, where $\sigma\|{\rm avg}^i\|_2\leq\sigma$. By standard tail bounds for the normal distribution, with probability at least $1-\beta$, in all of the iterations we have that $\langle {\rm avg}^i,\eta^i \rangle\leq\frac{1}{4}$.

    \item $\langle \eta^i,\eta^i \rangle$ is the sum of the squares of $d$ independent samples from $\NNN(0,\sigma^2)$, i.e., $\frac{1}{\sigma^2}\langle \eta^i,\eta^i \rangle$ is distributed according to the chi-squared distribution with $d$ degrees of freedom. By standard tail bounds for the chi-squared distribution, for every $y$ it holds that 
$\Pr[\frac{1}{\sigma^2}\langle \eta^i,\eta^i \rangle \geq d+2\sqrt{dy}+2y ]\leq\exp(-y)$. Thus, by our bound on $\sigma$, with probability at least $(1-\beta)$, in all of the $O(d^2)$ iterations we have that 
$\langle \eta^i,\eta^i \rangle \leq\frac{1}{4}$.

\item $\langle x^i,\eta^i \rangle$ is distributed as $\NNN(0,\sigma^2\|x^i\|^2_2)$. Hence, with probability at least $1-\beta$, for every iteration $1\leq i\leq R=\Theta(d^2)$ we have that $\langle x^i,\eta^i \rangle\leq\sigma\|x^i\|_2\sqrt{2\ln(R/\beta)}\leq\frac{\|x^i\|_2}{12000d}$, where the inequality follows by our bound on $\sigma$.
\end{itemize}
So, with probability at least $1-O(\beta)$ we have
\begin{eqnarray*}
(\|x\|_2)^2 &\leq& \langle x^i,x^i \rangle + 2 +  \frac{2\|x\|_2}{6000d}\\
&\leq&\langle x^i,x^i \rangle + 3 +  \left(\frac{2\|x\|_2}{6000d}\right)^2\\
&=& \langle x^i,x^i \rangle + 3 +  \frac{\|x\|^2_2}{9000000d}, 
\end{eqnarray*}
where the second inequality follows from the fact that $\forall v\in\R$ we have $v\leq1+v^2$. 
All in all, with probability at least $1-O(\beta)$, for every iteration $i$ we have that 
$$\langle x^{i+1},z \rangle \geq \langle x^i,z \rangle + \frac{1}{1000d} \qquad \text{and} \qquad \langle x^{i+1},x^{i+1} \rangle \leq \left(1+\frac{1}{9000000d^2}\right)\langle x^i,x^i \rangle + 3 $$
Hence, for every $i\leq9000000d^2$ we have that
$$\langle x^i,z \rangle \geq \langle x^0,z \rangle + \frac{i}{1000d} \qquad \text{and} \qquad \|x^i\|=\sqrt{\langle x^i,x^i \rangle} \leq \sqrt{e\langle x^0,x^0 \rangle + 3ie} $$
After $9000000d^2$ iterations we get that $\frac{\langle x,z \rangle}{\|x\|_2}>1$, which cannot happen.
\end{proof}

\begin{lem}[\cite{BFKV98}] \label{lem:bfkv98}
Let $A$ be an $n\times d$ matrix with rows $a_1,\ldots,a_n$, and let $z$ be any unit vector such that $Az\geq0$. 
Let $\nu=   \Theta\left(\frac{1}{\epsilon}d^{2.5}\ln(d)\ln(\frac{1}{\beta\delta})\right)$.
Then, applying the perceptron improvement phase of Algorithm \ref{alg:privateLP} on $A$ returns a vector $y$ such that
\begin{enumerate}
	\item W.p. $1-\beta$, whenever the algorithm breaks in Step 2(b)i it holds that $|\{i: \langle \bar{a}_i,\bar{y} \rangle<-\Delta\}|\leq \nu$.
	\item W.p. $\frac{1}{8}-\beta$, the algorithm breaks in Step~2(b)i and returns a vector $y$ such that  $\langle z,\bar{y} \rangle\geq\frac{1}{2\sqrt{d}}$.
\end{enumerate}
\end{lem}

\begin{proof}
We assume that every sample from $\Lap(1/\epsilon)$ in Step~2b(i) satisfies $\Lap(1/\epsilon)\leq
O\left(\frac{1}{\epsilon}\ln(\frac{d}{\beta})\right)\ll\nu=
       \Theta\left(\frac{1}{\epsilon}d^{2.5}\ln(d)\ln(\frac{1}{\beta\delta})\right)$. Since there are at most $\frac{8}{\Delta^2}\ln(3\sqrt{d}) = O(d^2\log d)$ iterations, this happens with probability at least $1-\beta$. In particular, if the algorithm breaks in Step~2b(i) then we obtained a vector $y$ s.t.\ $|\{i: \langle \bar{a_i},\bar{y} \rangle<-\Delta\}|\leq O(\nu).$

Let $y^i$ denote the vector $y$ as it is after the $i^\text{th}$ update of step~2b(iii), and let $u^i$ be the corresponding vector $u$ s.t.\ $y^{i+1}=y^i-\langle u^i,y \rangle u^i$.
Also let $m^i=|\{a\in A : \langle \bar{a},\bar{y}^i \rangle<-\Delta\}|$, and ${\rm avg}^i=\left(\frac{1}{m^i}\sum_{a\in A: \langle \bar{a},\bar{y}^i \rangle<-\Delta}\bar{a}\right)$. 
By our assumption on the noise magnitude in Step~2b(i), in every iteration that does not break in Step~2b(i) we have that $m^i\geq\Omega(\nu)$. Therefore, by the properties of algorithm \texttt{NoisyAVG}, with probability at least $1-\beta$, in all of those iterations we also have that $u^i={\rm avg}^i + \eta^i$ where $\eta^i$ is a vector whose every coordinate is sampled i.i.d.\ from $\NNN(0,\sigma^2)$ for some $\sigma\leq O\left(\frac{1}{\epsilon m^i}\sqrt{\ln(1/\delta)}\right)\leq O\left(\frac{1}{d^{2.5}\ln(d)\sqrt{\ln(\frac{1}{\beta})}}\right)$. We proceed with the analysis assuming that this is the case.

Let $z$ be any unit vector such that $Az\geq0$.
A standard computation shows that for a random unit vector $y$ it holds that $\langle z,y \rangle\geq\frac{1}{\sqrt{d}}$ with probability at least $1/8$. We continue with the analysis assuming that this is the case. By induction on $i$, we now show that (w.h.p.) for every $i$ we have that
\begin{enumerate}[leftmargin=100pt]
	\item[(1)] $\quad \langle y^i,y^i \rangle \leq \left(1 - \frac{\Delta^2}{8}\right)^i $
	\item[(2)] $\quad \langle y^i,z \rangle \geq \frac{1}{\sqrt{d}}-i\sigma\sqrt{8\ln(1/\beta)} $
\end{enumerate}

The base case follows since $y^0$ is a unit vector, so $\langle y^0,y^0 \rangle=1$, and by our assumption on $y^0$ and $z$ we have that $\langle y^0,z \rangle \geq \frac{1}{\sqrt{d}}$.
Now assume that (1) and (2) hold in the $i^{\text{th}}$ iteration. 
We first analyze $\langle y^{i+1},z \rangle$:
$$
\langle y^{i+1},z \rangle = \langle y^i-\langle y^i,u^i \rangle u^i ,z \rangle = \langle y^i ,z \rangle - \langle y^i ,u^i \rangle \langle  u^i,z \rangle
\geq \frac{1}{\sqrt{d}}-i\sigma\sqrt{8\ln(1/\beta)} - \langle y^i ,u^i \rangle \langle  u^i,z \rangle
$$
Note that 
$$\langle y^i,u^i \rangle = \langle y^i,{\rm avg}^i \rangle + \langle y^i,\eta^i \rangle = \|y^i\|_2\cdot \langle \bar{y}^i,{\rm avg}^i \rangle + \langle y^i,\eta^i \rangle < 
-\|y^i\|_2\cdot\Delta  + \langle y^i,\eta^i \rangle
,$$
\medskip
since $\langle \bar{y}^i,{\rm avg}^i \rangle<-\Delta$ by construction. Observe that $\langle y^i,\eta^i \rangle\sim\NNN(0,\|y^i\|_2^2\sigma^2)$, and hence with probability at least $1-\beta$ we have that 
\begin{equation}\label{eq:yiui}
\langle y^i,u^i \rangle < - \|y^i\|_2 \cdot \Delta  + \sigma\|y^i\|_2\sqrt{2\ln(1/\beta)} 
\leq
-\frac{\Delta}{2}\cdot\|y^i\|_2,
\end{equation}
where the last inequality follows by asserting that $\sigma\leq \frac{\Delta}{2\sqrt{2\ln(1/\beta)}}$. 
So $\langle y^i,u^i \rangle$ is negative. 
Hence, if $\langle  u^i,z \rangle$ is positive then we get that $\langle y^{i+1},z \rangle\geq\langle y^i ,z \rangle$ as desired.
Otherwise, assuming that $\langle  u^i,z \rangle$ is negative, we upper bound $|\langle y^i,u^i \rangle\cdot \langle u^i,z \rangle|$.
First note that $|\langle y^i,u^i \rangle|\leq2$ because $|\langle y^i,{\rm avg}^i \rangle|\leq1$ (since $\|y^i\|_2\leq1$ and $\|{\rm avg}^i\|_2\leq1$) and because $|\langle y^i,\eta^i \rangle|\leq1$ w.p. $1-\beta$.
 We now show that $\langle  u^i,z \rangle$ cannot be ``too negative'':
$\langle u^i,z \rangle = \langle {\rm avg}^i,z \rangle + \langle \eta^i,z \rangle \geq \langle \eta^i,z \rangle$, where the last inequality is because ${\rm avg}^i$ is the average of vectors $a$ satisfying $\langle a,z \rangle\geq0$. Observe that $\langle \eta^i,z \rangle\sim\NNN(0,\sigma^2)$, and therefore, with probability at least $1-\beta$ we have that 
$\langle u^i,z \rangle\geq-\sigma\sqrt{2\ln(1/\beta)}$. In this case, we have that
$\langle y^{i+1},z \rangle \geq \frac{1}{\sqrt{d}}-(i+1)\sigma\sqrt{8\ln(1/\beta)}$. So property (2) holds w.p.\ $1-\beta$.

We next analyze $\langle y^{i+1},y^{i+1} \rangle$:
\begin{eqnarray*}
\langle y^{i+1},y^{i+1} \rangle &=& \langle \left(y^i-\langle u^i,y^i \rangle u^i\right),\left(y^i-\langle u^i,y^i \rangle u^i\right) \rangle \\
&=& \langle y^i,y^i \rangle - 2\langle u^i,y^i \rangle^2 + \langle u^i,y^i \rangle^2\langle u^i,u^i \rangle
\end{eqnarray*}
Note that $\langle u^i,u^i \rangle=\langle {\rm avg}^i,{\rm avg}^i \rangle + 2\langle {\rm avg}^i,\eta^i \rangle + \langle \eta^i,\eta^i \rangle\leq 1 + 2\langle {\rm avg}^i,\eta^i \rangle + \langle \eta^i,\eta^i \rangle$, where the last inequality is because $\langle {\rm avg}^i,{\rm avg}^i \rangle\leq1$ by the triangle inequality (${\rm avg}^i$ is the average of unit vectors). Observe that $\langle {\rm avg}^i,\eta^i \rangle \sim\NNN(0,\|{\rm avg}^i\|_2^2\sigma^2)=\NNN(0,\sigma'^2)$ for some $\sigma'\leq\sigma$, and that $\frac{1}{\sigma^2}\langle \eta^i,\eta^i \rangle$ is distributed according to the chi-squared distribution  with $d$ degrees of freedom. Hence, with probability at least $1-2\beta$ we get that
$\langle u^i,u^i \rangle\leq1+2\sigma\sqrt{2\ln(1/\beta)}+5d\sigma^2\ln(1/\beta)<\frac{3}{2}$. So,
\begin{eqnarray*}
\langle y^{i+1},y^{i+1} \rangle &\leq& \langle y^i,y^i \rangle - \frac{1}{2}\langle u^i,y^i \rangle^2,
\end{eqnarray*}
and by Inequality~(\ref{eq:yiui}),
\begin{eqnarray*}
\langle y^{i+1},y^{i+1} \rangle &\leq& \langle y^i,y^i \rangle - \frac{1}{2}\cdot\left(\frac{\Delta}{2}\cdot\|y^i\|_2\right)^2
=
\|y^i\|_2^2\cdot\left( 1 - \frac{\Delta^2}{8} \right).
\end{eqnarray*}
So property (1) holds with high probability.\\

Using property (1), after $t=\frac{8}{\Delta^2}\ln(3\sqrt{d})$ steps we get that $\|y^t\|_2\leq\frac{1}{3\sqrt{d}}$. Using property (2), asserting that $\sigma\leq \frac{1}{2 t \sqrt{8d\ln(1/\beta)} }$, after $t$ steps we also have that $\langle y^i,z \rangle \geq \frac{1}{2\sqrt{d}}$, so $\frac{\langle y^i,z \rangle}{\|y^t\|_2}>1$, which cannot happen. Thus, assuming all the above mentioned events occur, the algorithm must break and return a vector $y$ s.t.\ $|\{i: \langle \bar{a_i},\bar{y} \rangle<-\Delta\}|\leq O(\nu)$. In addition, for every $i$ we have that 
\begin{eqnarray*}
\langle \bar{y}^i,z \rangle &=& \frac{1}{\|y^i\|_2} \langle y^i,z \rangle \\
&\geq& \left(1 - \frac{\Delta^2}{8}\right)^{-i/2} \left( \frac{1}{\sqrt{d}}-i\sigma\sqrt{8\ln(1/\beta)} \right)\\
&\geq& \frac{1}{2\sqrt{d}}.
\end{eqnarray*}
\end{proof}

\begin{obs}\label{obs:beta8th}
Let $A$ be an $n\times d$ matrix with rows $a_1,\ldots,a_n$, and let $z$ be any unit vector such that $Az\geq0$. 
Let $\nu=     \Theta\left(\frac{1}{\epsilon}d^{2.5}\ln(d)\ln(\frac{1}{\beta\delta})\right)$ as in the proof of Lemma \ref{lem:bfkv98}.
Then, running Step~2 on $A$ results in a vector $y$ such that
\begin{enumerate}
	\item W.p. $1-\beta$ we have that $|\{i: \langle \bar{a}_i,\bar{y} \rangle<-\Delta\}|\leq \nu$.
	\item W.p. $\frac{1}{8}-\beta$ we have that $\langle z,\bar{y} \rangle\geq\frac{1}{2\sqrt{d}}$.
\end{enumerate}
\end{obs}

Recall that in Step~3 we delete from $A$ every row $a_i$ such that $\langle \bar{a}_i,\bar{y} \rangle<-\Delta$. The above observation shows that w.h.p.\ the number of deleted rows is small.

\begin{lem}
Suppose that $\rho,\Delta\leq\frac{1}{500d}$. Let $A'$ be obtained from $A$ by one
iteration of the algorithm (one on which the problem was not solved). Let $\rho'$ and $\rho$ denote the roundness of $A'$ and $A$, respectively. Then,
\begin{enumerate}
    \item[(a)] $\rho'\geq\left(1-\frac{1}{498d}\right)\rho$.
    \item[(b)] With probability at least $1/8-\beta$ we have $\rho'\geq\left(1+\frac{7}{75d}\right)\rho$.
\end{enumerate}
\end{lem}

\begin{proof}

Let $a_i$, $i=1,\dots,n$ be the rows of $A$ at the beginning of some iteration. 
Let $z$ be a unit vector satisfying $\rho=\min_{i}  \langle \bar{a_i},z \rangle$, and let $y$ be the vector obtained after Step~3 of the algorithm. For $i\in[n]$ let $\Delta_i=\langle \bar{a_i},\bar{y} \rangle$. Recall that some of $A$'s rows might get deleted in Step~3, and let us assume for simplicity that the remaining rows are $1,2,\dots,n'$ for some $n'\leq n$.

 Using these notations, after Step~3, for every $i\in[n']$ it holds that
$$
\langle \bar{a_i},\bar{y} \rangle = \Delta_i \geq -\Delta.
$$

As in the statement of the lemma, let $A'$ be the matrix obtained after the rescaling step. That is, $A'$ is the matrix with $n'$ rows of the form
$$
a'_i=a_i + \langle a_i,\bar{y} \rangle\cdot \bar{y}.
$$
Denote $\gamma=\frac{1}{2}(\rho-\langle \bar{y},z \rangle)$,
and define
$$
z'=z+\gamma\bar{y}.
$$
We have
\begin{equation}\label{eq:rho-tag}
\rho'=
\max_{\substack{x:\|x\|_2=1,\\A'x\geq0}} \;\min_{i\in[n']} \; \langle \bar{a'_i},x \rangle\geq 
\min_{i\in[n']} \; \langle \bar{a'_i},\bar{z'} \rangle 
= \min_{i\in[n']}  \frac{\langle \bar{a'_i}, z' \rangle }{\|z'\|}.
\end{equation}
To lower bound $\rho'$, we first show that $\langle \bar{a'_i}, z' \rangle$ cannot be too small:\footnote{
Note that 
$\|\bar{a}_i + \langle \bar{a}_i,\bar{y} \rangle\cdot \bar{y}\|=\sqrt{\Big\langle\bar{a}_i + \langle \bar{a}_i,\bar{y} \rangle\cdot \bar{y} \;\;,\;\; \bar{a}_i + \langle \bar{a}_i,\bar{y} \rangle\cdot \bar{y}\Big\rangle}=\sqrt{
\langle \bar{a}_i,\bar{a}_i \rangle
+2\langle \bar{a}_i,\bar{y} \rangle^2
+\langle \bar{a}_i,\bar{y} \rangle^2\cdot \langle \bar{y},\bar{y} \rangle}
=\sqrt{
1
+3\langle \bar{a}_i,\bar{y} \rangle^2}$
}
\begin{eqnarray*}
\langle \bar{a'_i}, z' \rangle &=& 
\left\langle \left(\frac{a_i + \langle a_i,\bar{y} \rangle\cdot \bar{y}}{\|a_i + \langle a_i,\bar{y} \rangle\cdot \bar{y}\|}\right) , z' \right\rangle \\[1em] 
&=& \left\langle \left(\frac{\bar{a}_i + \langle \bar{a}_i,\bar{y} \rangle\cdot \bar{y}}{\|\bar{a}_i + \langle \bar{a}_i,\bar{y} \rangle\cdot \bar{y}\|}\right) , z' \right\rangle \\
&=&\frac{ \Big\langle \bar{a}_i + \langle \bar{a}_i,\bar{y} \rangle\cdot \bar{y} \;\;,\;\; z+\gamma\bar{y} \Big\rangle}{\sqrt{
1
+3\langle \bar{a}_i,\bar{y} \rangle^2}}\\[1em] 
&=&\frac{ 
\langle \bar{a}_i, z \rangle
+
\gamma\langle \bar{a}_i,\bar{y} \rangle
+
\langle \bar{a}_i,\bar{y} \rangle\langle \bar{y},z \rangle
+
\gamma\langle \bar{a}_i,\bar{y} \rangle\langle \bar{y},\bar{y} \rangle
}{\sqrt{
1
+3\Delta_i^2}}\\[1em] 
&=&\frac{ 
\langle \bar{a}_i, z \rangle
+
\langle \bar{a}_i,\bar{y} \rangle\langle \bar{y},z \rangle
+
2\gamma\langle \bar{a}_i,\bar{y} \rangle
}{\sqrt{
1
+3\Delta_i^2}}\\[1em] 
&\geq&\frac{ 
\rho
+
\langle \bar{a}_i,\bar{y} \rangle\langle \bar{y},z \rangle
+
2\gamma\langle \bar{a}_i,\bar{y} \rangle
}{\sqrt{
1
+3\Delta_i^2}}\\[1em] 
&=&\frac{ 
\rho
+
\langle \bar{a}_i,\bar{y} \rangle\langle \bar{y},z \rangle
+
(\rho-\langle \bar{y},z \rangle)\langle \bar{a}_i,\bar{y} \rangle
}{\sqrt{
1
+3\Delta_i^2}}
\\[1em] 
&=&\rho\cdot\frac{ 
1
+
\langle \bar{a}_i,\bar{y} \rangle}{\sqrt{
1
+3\Delta_i^2}}
=\rho\cdot\frac{ 
1
+
\Delta_i}{\sqrt{
1
+3\Delta_i^2}}
\geq\rho\cdot\frac{ 
1
-
\Delta}{\sqrt{
1
+3\Delta^2}},
\end{eqnarray*}
where the last inequality follows from from $\Delta_i\in[-\Delta,1]$ and from the fact that the minimum of the function $f(x)=\frac{1+x}{\sqrt{1+3x^3}}$ over $x\in[-\Delta,1]$ is obtained for $x=-\Delta$.
Plugging this in Inequality~(\ref{eq:rho-tag}) we get that
\begin{equation}\label{eq:rho-tag2}
\rho'\geq\rho\cdot\frac{ 
1
-
\Delta}{\sqrt{
1
+3\Delta^2}\cdot\|z'\|}.
\end{equation}
We next show that $\|z'\|$ cannot be too large. First observe that
\begin{eqnarray}\label{eq:z-tag}
\|z'\|^2=\|z+\gamma\bar{y}\|^2=\|z\|^2+\gamma^2\|\bar{y}\|^2+2\gamma \langle \bar{y},z \rangle
=1+\gamma^2+2\gamma\langle \bar{y},z \rangle
=1+\frac{\rho^2}{4}+\frac{\rho}{2}\langle \bar{y},z \rangle-\frac{3}{4}\langle \bar{y},z \rangle^2,
\end{eqnarray}
where the last equality follows from the definition of $\gamma$. Now, viewing $\|z'\|^2$ as a quadratic polynomial in $\langle \bar{y},z \rangle$, we see that it is maximized when $\langle \bar{y},z \rangle=\frac{\rho}{3}$. This shows that 
$$
\|z'\|^2\leq1+\frac{\rho^2}{3}.
$$
Plugging this into Inequality~(\ref{eq:rho-tag2}) we get that
$$
\rho'\geq\rho\cdot\frac{ 
1
-
\Delta}{\sqrt{
1
+3\Delta^2}\cdot\sqrt{1+\frac{\rho^2}{3}}}.
$$
Using the elementary inequality $\frac{1}{\sqrt{1+x}}\geq1-\frac{x}{2}$ for all $x>-1$ we get
\begin{eqnarray*}
\rho'&\geq& \rho(1-\Delta)\left(1-\frac{3\Delta^2}{2}\right)\left(1-\frac{\rho^2}{6}\right)\\[1em]
&\geq&\rho\left(1-\Delta-\frac{3\Delta^2}{2}-\frac{\rho^2}{6}\right)\\[1em]
&\geq&\rho\left(1-\frac{1}{498d}\right),
\end{eqnarray*}
where the last inequality follows since $\Delta,\rho\leq\frac{1}{500d}$. This establishes Part (a) of the lemma: $\rho'$ is never too small compared to $\rho$. 

Next, we establish Part (b) of the lemma. That is, we need to show that with constant probability $\rho'$ is in fact noticeably larger than $\rho$. To this end, recall that by Observation~\ref{obs:beta8th}, with probability at least $\frac{1}{8}-\beta$ it holds that $\langle \bar{y},z \rangle\geq\frac{1}{2\sqrt{d}}$. In this case, by Equality~(\ref{eq:z-tag}) we have
$$
\|z'\|^2=1+\frac{\rho^2}{4}+\frac{\rho}{2}\langle \bar{y},z \rangle-\frac{3}{4}\langle \bar{y},z \rangle^2\leq 
1+\frac{\rho^2}{4}+\frac{\rho}{4\sqrt{d}}-\frac{3}{16d},
$$
where the last inequality is due to the fact that the polynomial $p(x)=1+\frac{\rho^2}{4}+\frac{\rho}{2}x-\frac{3}{4}x^2$ is decreasing for every $x\geq\frac{\rho}{3}$, together with the assumption that $\langle \bar{y},z \rangle\geq\frac{1}{2\sqrt{d}}>\frac{1}{3\cdot 32d}\geq\frac{\rho}{3}$. Plugging this into Inequality~(\ref{eq:rho-tag2}) and using again the elementary inequality mentioned above, we have
\begin{eqnarray*}
\rho'&\geq& \rho(1-\Delta)\left(1-\frac{3\Delta^2}{2}\right)\left(1-\frac{\rho^2}{8}-\frac{\rho}{8\sqrt{d}}+\frac{3}{32d}\right)\\[1em]
&\geq&\rho\left(1+\frac{7}{75d}\right),
\end{eqnarray*}
where the last inequality follows from the assumption that $\Delta,\rho\leq\frac{1}{500d}$.
This proves both parts of the lemma.
\end{proof}

We can now bound the number of iterations. Let $A_0$ denote the initial input matrix and let $\rho_0 = \rho(A_0)$ denote the roundness of its feasible region. 

\begin{lem}\label{lem:boundIterations}
Consider applying Algorithm~\ref{alg:privateLP} on an input matrix $A$ satisfying $\rho(A)\geq\rho_0$.
With probability at least $1-O(\beta)$ the perceptron phase succeeds in at most $T=\Omega(d\cdot\ln(1/\rho_0)+\ln(1/\beta))$ iterations.  
\end{lem}

\begin{proof}
It suffices to show that with probability at least $1-\beta$, in at most $T$ iterations we  have that $\rho\geq\frac{1}{500d}$. Let $X_i$ be a random variable
for the $i$'th iteration, with value $0$ if $\rho$ grows by a factor of $(1 + \frac{7}{75d})$ or more
and value 1 otherwise. Strictly speaking, these random variables are not i.i.d. Still, for every fixture of $X_1,\dots,X_{i-1}$ it holds that $\Pr[X_i|X_1,\dots,X_{i-1}]\leq\frac{7}{8}+\beta$. Thus, by the Azuma–Hoeffding inequality we have
$$
\Pr\left[\sum_{i=1}^T X_i\geq T\left(\frac{7}{8}+\beta\right)+z\sqrt{T}\right]\leq\exp\left(-\frac{z^2}{2}\right).
$$
Choosing $z=\frac{\sqrt{T}}{256}$ and asserting that $\beta\leq\frac{1}{256}$ this gives
$$
\Pr\left[\sum_{i=1}^T X_i\geq \frac{7T}{8}+\frac{T}{128}\right]\leq\exp\left(-\frac{T}{2\cdot 256^2}\right).
$$
This probability is at most $\beta$ whenever $T\geq\Omega(\ln(1/
\beta))$. 
Analyzing $\rho_T$ in the case that $\sum_{i=1}^T X_i\leq \frac{7T}{8}+\frac{T}{128}$ we have
\begin{eqnarray}
\rho_T&\geq&\rho_0\left(1+\frac{7}{75d}\right)^{\frac{T}{8}-\frac{T}{128}}\left(1-\frac{1}{498d}\right)^{\frac{7T}{8}+\frac{T}{128}}\\
&\geq&\rho_0\cdot e^{\frac{7}{150d}\cdot\frac{15T}{128}}
\cdot e^{-\frac{2}{498d}\cdot\frac{113T}{128}}\\
&\geq&\rho_0\cdot e^{\frac{T}{1000d}}.
\end{eqnarray}
This is at least $\frac{1}{500d}$ whenever $T=\Omega(d\cdot\ln(1/\rho_0))$. In summary, with probability at least $1-\beta$, in at most $T=\Omega(d\cdot\ln(1/\rho_0)+\ln(1/\beta))$ iterations, $\rho$ grows to at least $\frac{1}{500d}$, at which point the perceptron phase succeeds with probability at least  $1-\beta$.
\end{proof}

\begin{theorem}
Let $n> N=\Theta\left(\frac{d^{3.5}}{\epsilon}\ln(d)\ln(\frac{1}{\beta})\ln(\frac{1}{\beta\delta})\ln(\frac{1}{\rho_0})\right)$ and denote $T=\Theta(d\cdot\ln(1/\rho_0)+\ln(1/\beta)) = O(d\ln(1/\rho_0)\ln(1/\beta))$.
Let $A_0$ be an $n\times d$ matrix with rows $a_1,\ldots,a_n$ such that $\rho(A)\geq\rho_0$, and consider applying Algorithm \ref{alg:privateLP} on $A_0$. Then with probability at least $1-O(\beta\cdot T)$ the algorithm halts during its first $T$ iterations and returns a vector $x^*$ such that $|\{i: \langle \bar{a}_i,\bar{x}^* \rangle< \frac{\Delta}{24}\cdot 2^{-T} \}|\leq N$.
\end{theorem}

\begin{proof}
By Lemma~\ref{lem:boundIterations}, with probability at least $1-O(\beta)$ the perceptron phase succeeds during the first $T$ iterations of the algorithm. When it succeeds, it returns a vector $Bx^*$ such that there are at most $2\zeta=\Theta\left(\frac{d^2}{\epsilon}\ln(\frac{1}{\beta\delta})\right)$ rows $a_i$ in the {\em current matrix} for which
$\langle \bar{a}_i,\bar{x}^* \rangle<\frac{\Delta}{24}$. There are two differences between the current matrix and the original input matrix:
\begin{enumerate}
    \item Throughout the execution, some rows are deleted from the matrix (in Step~3). By Observation~\ref{obs:beta8th}, with probability at least $1-\beta\cdot T$, in each iteration we delete at most $\nu=
       \Theta\left(\frac{1}{\epsilon}d^{2.5}\ln(d)\ln(\frac{1}{\beta\delta})\right)$ rows from the matrix. So overall there could be at most $\nu T$ rows that are missing from the current matrix.

\item The rows of the input matrix get rescaled with each iteration (in Step~6).
\end{enumerate}

To address Item~2 above, we will show that if $x^*$ satisfies a ``rescaled constraint'' from the current matrix then $Bx^*$ satisfies the corresponding ``original'' constraint from the input matrix (albeit with an exponentially smaller ``wiggle room''). Let $A_0$ denote the input matrix and let $A_t$ denote the matrix obtained after the $t$th rescaling step. Also let $B_t=\left(I+\bar{y}_t\bar{y}_t^T\right)$, where $y_t$ is the vector used for the $t$th rescaling step. Note that with these notations, the matrix B resulting from the $t$th rescaling iteration (in Step~6) can be written as $B=B_1B_2\cdots B_t$.

Now fix $t$, let $a$ be a row from $A_t$, and and let $a^{\rm rescaled}=a+\langle a,\bar{y}_t \rangle\bar{y}_t$ be its corresponding ``rescaled row'' in $A_{t+1}$. 
First observe that $\|a^{\rm rescaled}\|_2\geq\|a\|_2$, as can be seen by the following calculation:

\begin{eqnarray*}
\langle a^{\rm rescaled},a^{\rm rescaled} \rangle
&=& \Big\langle a+\langle a,\bar{y}_t \rangle\bar{y}_t\;,\;a+\langle a,\bar{y}_t \rangle\bar{y}_t \Big\rangle\\
&=& \langle a,a \rangle + 2\cdot \langle a,\bar{y}_t \rangle^2
+ \langle a,\bar{y}_t \rangle^2 \cdot \langle \bar{y}_t,\bar{y}_t \rangle^2\\
&=& \langle a,a \rangle + 3\cdot \langle a,\bar{y}_t \rangle^2\\
&\geq& \langle a,a \rangle.
\end{eqnarray*}

We next show that for any vector $v\in\R^d$ it holds that 
$\left\langle a, B_t v\right\rangle=\left\langle a^{\rm rescaled},v\right\rangle$. To this end, note that $B_t$ is the matrix whose $\ell$th row is 
    \begin{eqnarray*}
    b_{\ell}&=&   e_{\ell}+\bar{y}_t[\ell]\cdot\bar{y}_t.
    \end{eqnarray*}
    Thus, $B_t \, v$ is the vector whose $\ell$th entry is
    $$
    v[\ell]+\bar{y}_t[\ell]\cdot\langle \bar{y}_t,v \rangle.
    $$
    And so,
    \begin{eqnarray*}
    \langle a , B_t\, v \rangle &=& 
    a[1]\Big(v[1]+\bar{y}_t[1]\cdot\langle \bar{y}_t,v \rangle\Big)+\dots+
    a[d]\Big(v[d]+\bar{y}_t[d]\cdot\langle \bar{y}_t,v \rangle\Big)\\[1em]
    &=&\langle a,v \rangle+
    \langle a,\bar{y}_t \rangle\cdot\langle \bar{y}_1,v \rangle\\[1em]
    &=&\left\langle \Big(a+\langle a,\bar{y}_t \rangle\bar{y}_t\Big) ,v \right\rangle
    =
    \langle a^{\rm rescaled},v\rangle.
    \end{eqnarray*}
Furthermore,
\begin{eqnarray*}
 \left\langle \frac{a}{\|a\|_2} , \frac{B_t\, v}{\|B_t\, v\|_2} \right\rangle 
    &=&\left\langle \frac{a}{\|a\|_2} , \frac{B_t\, \bar{v}}{\|B_t\, \bar{v}\|_2} \right\rangle \\[1em]
    &=& \frac{ \langle a^{\rm rescaled},\bar{v}\rangle}{\|a\|_2\cdot\|B_t\, \bar{v}\|_2}\\[1em]
    &\geq&
    \frac{ \langle a^{\rm rescaled},\bar{v}\rangle}{\|a^{\rm rescaled}\|_2\cdot\left(\|\bar{v}\|_2 + \langle \bar{y}_t,\bar{v} \rangle \|\bar{y}_t\|\right)}\\[1em]
    &\geq&
    \frac{ \langle a^{\rm rescaled},\bar{v}\rangle}{\|a^{\rm rescaled}\|_2\cdot2}
    = \frac{1}{2}\cdot \left\langle \frac{a^{\rm rescaled}}{\|a^{\rm rescaled}\|_2} , \frac{ v}{\|v\|_2} \right\rangle   
\end{eqnarray*}

This shows that every rescaling operation loses at most a factor 2 in the resulting ``wiggle room''. Specifically, let $x^*$ be the vector identified by a successful iteration of the perceptron phase, say w.r.t.\ the matrix $A_t$ for some $t\leq T$. Let $a^{(t)}$ be a row in $A_t$ for which $\langle a^{(t)},\bar{x}^* \rangle\geq \frac{\Delta}{24}$, and let $a^{(0)},a^{(1)},\dots,a^{(t-1)}$ denote the corresponding ``unscaled'' rows in $A_0,A_1,\dots,A_{t-1}$, respectively. Then,

\begin{eqnarray*}
\left\langle \frac{a^{(0)}}{\|a^{(0)}\|_2} , \frac{B x^*}{\|B x^*\|_2} \right\rangle 
    &=& \left\langle \frac{a^{(0)}}{\|a^{(0)}\|_2} , \frac{B_1\cdots B_t x^*}{\|B_1\cdots B_t x^*\|_2} \right\rangle \\
    &\geq&
     \frac{1}{2}\cdot \left\langle \frac{a^{(1)}}{\|a^{(1)}\|_2} , \frac{B_2\cdots B_t x^*}{\|B_2\cdots B_t x^*\|_2} \right\rangle \geq\cdots
     \geq
     \frac{1}{2^t}\cdot \left\langle \frac{a^{(t)}}{\|a^{(t)}\|_2} , \frac{x^*}{\|x^*\|_2} \right\rangle\\
     &\geq&2^{-T}\cdot \frac{\Delta}{24}.
\end{eqnarray*}
\end{proof}

\begin{theorem}
Denote 
$k=\Theta\left(d^{3}\ln(d)\ln^2(\frac{1}{\beta})\ln(\frac{1}{\rho_0})\right)$. 
Algorithm~\ref{alg:privateLP} satisfies $\left(\hat{\epsilon},2k\delta\right)$-differential privacy for 
$\hat{\epsilon}=\sqrt{2k\ln(\frac{1}{k\delta})}\cdot\epsilon + 2k\epsilon^2$.
\end{theorem}

\begin{proof}
Follows from standard composition theorems for differential privacy \citep{DRV10}, together with the fact that throughout the execution the algorithm accesses its input database via at most $k$ applications %
of \texttt{NoisyAVG}, each of which satisfies $(\epsilon,\delta)$-differential privacy.
\end{proof}

By appropriately rescaling the parameters $\beta,\epsilon,\delta$ in Algorithm~\ref{alg:privateLP} we get the following result.

\begin{theorem}
Denote $T=\Theta(d\cdot\ln(1/\rho_0)+\ln(1/\beta))$. 
There exists a computationally efficient $(\epsilon,\delta)$-differentially private algorithm whose input is an $n\times d$ matrix $A$ and its output is a vector $x^*\in\R^d$ such that the following holds. Let $a_1,\ldots,a_n$ denote the rows of $A$. If $\rho(A)\geq\rho_0$, then with probability at least $1-\beta$ the returned vector $x^*$ satisfies 
\begin{eqnarray*}
\left|\left\{i: \langle \bar{a}_i,\bar{x}^* \rangle<\frac{\Delta}{24}\cdot 2^{-T}\right\}\right|&\leq& O\left(
\frac{d^5}{\epsilon}
 \ln^{1.5}(d) \ln^2\left(\frac{d\ln(\frac{1}{\rho_0})}{\beta}\right)  \ln^{1.5}\left(\frac{1}{\rho_0}\right) \ln\left(\frac{d\ln(\frac{1}{\rho_0})}{\beta\delta}\right)
 \sqrt{\ln\frac{1}{\delta}}
 \right)\\
 &=&\frac{d^5}{\epsilon}\cdot{\rm polylog}\left(d,\frac{1}{\beta},\frac{1}{\delta},\frac{1}{\rho_0}\right).
\end{eqnarray*}
\end{theorem}

\section{When the LP is not fully dimensional}
\label{app:fulld}

\begin{lem} \label{lem:cramer}
Consider the linear program
\[ b^T y = -1 \]
\[ A^T y \ge 0 \]
\[ y_i\ge 0, \;\;  \forall i\in I \subset [m]  \]
where $b$ is an integer vector in ${\cal R}^m$,
$A$ is an $m\times d$ integer matrix, and $y$ is an $m$-dimensional vector of variables.
We assume that $m\gg d$.
Let $U$ be an upper bound on the absolute value of $b_i$ and $A_{ij}$ for $i\in [m]$ and  $j\in [d]$.
If this LP is feasible then it has a  feasible solution $z\in {\cal R}^m$ 
such that 
1) $z$ has at most $d+1$ nonzero coordinates, and (2) 
$|z_i|\le ((d+1)U)^{d+1}$ for all $i\in [m]$.
\end{lem}
\begin{proof}
We pick $z$ on the flat of minimum dimension on the boundary of the polyhedron defined by this LP.
This flat is defined by at most $m$ of the constraints, and it contains all points that satisfy these $m$ constraints with equality.\footnote{This flat is a vertex defined by $m$ constraints if the polyhedron has a vertex, and otherwise it may be some unbounded flat of higher dimension completely contained in the polyhedron.}
Since only $d+1$ of the constraints are not of the form 
$y_i\ge 0$ it follows
that $z$ has at most $d+1$ non-zero coordinates.

 By Cramer's law,
 each nonzero coordinate of $z$ is a  quotient  of  determinants of square submatrices of dimension at most $d+1$, consisting of numbers from $A$ and $b$. Since the entries in $A$ and $b$ are integers, the absolute value of the denominator of this quotient is at least $1$, Since the absolute value of these integers is at most $U$, the absolute value of the numerator is at most $(d+1)! U^{d+1} \le ((d+1)U)^{d+1}$. So the lemma follows.
\end{proof}

\begin{lem} \label{lem:Peta2}
The linear program
\begin{eqnarray*}
 A_1 x & \le & b_1   \\
 A_2 x & = & b_2 \;\;\; (P) \\
 x & \ge & 0  
\end{eqnarray*}
is feasible if the linear program 
\begin{eqnarray*}
A_1 x & \le & b_1 \\
 A_2 x & = & b_2 + \eta_2 \;\;\; (P^\eta) \\
 x & \ge& 0 
\end{eqnarray*}
is feasible.
Here $A_1$ is an $m_1\times d$ matrix, $A_2$ is an $m_2\times d$,
$b_1$ is a vector of length $m_1$, and $b_2$ and $\eta_2$ are vectors of length $m_2$. $x$ is a vector of $d$ variables. 
The entries of $A_1$, $A_2$, $b_1$, and $b_2$ are integers of absolute value at most $U$.
The entries of $\eta_2$ are positive and smaller than $\frac{1}{2(d+1)((d+1)U)^{d+1}}$.
 We assume that $d\ll m=m_1+m_2$.
\end{lem}
\begin{proof}
Suppose $(P)$ is not feasible.
Add to $(P)$ an artificial objective 
\[\max \;\; 0\cdot x   \]
The dual to $(P)$ is
\begin{eqnarray*}
 \lefteqn{ \hspace*{-0.2in} \min \;  b^T y}  \\
 A^T y & \ge & 0  \;\;\;  (D) \\
 y_1 & \ge & 0  
\end{eqnarray*}
where $A$ is a matrix of $m$ rows consisting of the $m_1$ rows of $A_1$ and then the $m_2$ rows of $A_2$,
$b$ is the concatenation of $b_1$ and $b_2$,
$y$ is a vector of $m$ dual variables, and $y_1$ is a vector of the $m_1$ first variables in $y$ (duals of the first $m_1$ inequalities in $(P)$).

Clearly $(D)$ is feasible ($y=0$ is a solution), and therefore, since we assume that $(P)$ is infeasible, then by LP duality $(D)$ must be unbounded.
It follows that the polyhedron defined by
\begin{eqnarray*}
  b^T y & = & -1  \\
 A^T y & \ge & 0   \\
 y_1 & \ge & 0  
\end{eqnarray*}
is not empty.

By Lemma \ref{lem:cramer}, there is a point $z$ on the boundary of the feasible region of this LP, s.t.\ 
(1) At most $d+1$ of the coordinates of $z$ are not zero, and (2)
$|z_i|\le ((d+1)U)^{d+1}$ for all $i\in [m]$. So from our definition of $\eta_2$ follows that
\[
\left(b + \binom{\bf 0}{\eta_2}\right)^T z =
b^T z + 1/2 \le -1/2 \ .
\]
Here ${\bf 0}$ is a zero vector of $m_1$ coordinates.
This implies that
\[
\left(b + \binom{\bf 0}{\eta_2}\right)^T c\cdot z =
c(b^T z + 1/2) \le -c/2 \ .
\]
 for any constant $c\ge 0$.
So we conclude that the dual 
$D^\eta$ of $P^\eta$ (with a zero objective added)
given by
\begin{eqnarray*}
 \lefteqn{ \hspace*{-0.2in} \min \;  \left(b + \binom{\bf 0}{\eta_2}\right)^T y }  \\
 A^T y & \ge & 0  \;\;\;  (D^\eta) \\
 y_1 & \ge & 0  
\end{eqnarray*}
is also unbounded. Thus $P^\eta$ is infeasible. 
\end{proof}

\bibliographystyle{plainnat}

\end{document}